\documentclass[11pt,a4paper]{article}
\usepackage[utf8]{inputenc}
\usepackage[T1]{fontenc}
\usepackage[a4paper]{geometry}
\usepackage{lmodern,textcomp,amsmath,amssymb,mathtools,braket,authblk,amsthm}
\usepackage[english]{babel}

\newtheorem{theorem}{Theorem}[section]
\newtheorem{corollary}[theorem]{Corollary}
\newtheorem{lemma}[theorem]{Lemma}
\newtheorem{proposition}[theorem]{Proposition}
\theoremstyle{definition}
\newtheorem{definition}[theorem]{Definition}

\newtheorem{hypothesis}[theorem]{Hypothesis}
\theoremstyle{remark}
\newtheorem{remark}[theorem]{Remark}

\newtheorem*{example}{Example}

\newcommand{\hilb}{\mathcal{H}}
\newcommand{\hilbn}{\mathcal{H}^{(n)}}
\newcommand{\fock}{\mathcal{F}}
\newcommand{\focks}{\mathcal{F}}
\newcommand{\hfrak}{\mathfrak{H}}

\newcommand{\e}{\mathrm{e}}
\renewcommand{\Im}{\operatorname{Im}}

\newcommand{\Psie}{\Psi_{\mathrm{e}}}
\newcommand{\Psig}{\Psi_{\mathrm{g}}}

\newcommand{\Phie}{\Phi_{\mathrm{e}}}
\newcommand{\Phig}{\Phi_{\mathrm{g}}}
\renewcommand{\a}[1]{a\!\left(#1\right)}
\newcommand{\adag}[1]{a^\dag\!\left(#1\right)}
\newcommand{\ii}{\mathrm{i}}
\newcommand{\dOmega}{\mathrm{d}\Gamma(\omega)}

\makeatletter
\def\smalloverbrace#1{\mathop{\vbox{\m@th\ialign{##\crcr\noalign{\kern3\p@}%
	\tiny\downbracefill\crcr\noalign{\kern3\p@\nointerlineskip}%
	$\hfil\displaystyle{#1}\hfil$\crcr}}}\limits}
\makeatother

\usepackage[unicode=true,
bookmarks=true,bookmarksopen=false,
breaklinks=false,pdfborder={0 0 0},colorlinks=true]
{hyperref}
\usepackage{xcolor}
\definecolor{cblue}{rgb}{0.16, 0.32, 0.75}
\definecolor{cred}{rgb}{0.7, 0.11, 0.11}
\hypersetup{%
	,linkcolor=cred
	,citecolor=cblue
	,urlcolor=black
}

\usepackage[
backend=bibtex,
style=numeric-comp,
giveninits=true,
maxbibnames=299,
natbib=true,
doi=false,
isbn=false,
url=false,
sorting=none
]
{biblatex}
\renewbibmacro{in:}{}
\title{\textbf{Generalized spin-boson models\\with non-normalizable form factors}}

\author[$1,2,\star$]{Davide Lonigro}

\affil[$1$]{\small Dipartimento di Fisica and MECENAS, Universit\`a di Bari, I-70126 Bari, Italy}
\affil[$2$]{\small INFN, Sezione di Bari, I-70126 Bari, Italy}
\affil[$\star$]{\small \texttt{davide.lonigro@ba.infn.it}}
\begin{document}
	
\maketitle

\begin{abstract}
Generalized spin-boson (GSB) models describe the interaction between a quantum mechanical system and a structured boson environment, mediated by a family of coupling functions known as form factors. We propose an extension of the class of GSB models which can accommodate non-normalizable form factors, provided that they satisfy a weaker growth constraint, thus accounting for a rigorous description of a wider range of physical scenarios; we also show that such ``singular'' GSB models can be rigorously approximated by GSB models with normalizable form factors. Furthermore, we discuss in greater detail the structure of the spin-boson model with a rotating wave approximation (RWA): for this model, the result is improved via a nonperturbative approach which enables us to further extend the class of admissible form factors, as well as to compute its resolvent and characterize its self-adjointness domain.
\end{abstract}

\section{Introduction}

The spin-boson model, which describes the interaction between a quantum mechanical two-state system (qubit) and a structured boson environment, is one of the cornerstone of physics~\cite{weiss2012quantum,leggett1987dynamics}. Apart from providing a comprehensive and tractable description of fundamental phenomena, such as quantum noise, decoherence and non-Markovianity in open quantum systems~\cite{breuer2002theory,ingold2002path,grifoni1999dissipation,thorwart2004dynamics,clos2012quantification,costi2003entanglement}, it finds applications in a wide range of topics, such as quantum optics~\cite{bloch2012quantum,blatt2012quantum,porras2008mesoscopic,vogel1995nonlinear,phoenix1991establishment}, quantum information and simulation~\cite{hao2013dynamics,ge2010quantum,lemmer2018trapped,leppakangas2018quantum}, solid state and chemical physics~\cite{wipf1987influence,suarez1991hydrogen,guinea1985bosonization}. The interest in such models is increasingly fostered by the recent breakthrough in quantum technology: complex high-dimensional quantum systems can now be inspected and controlled with an unprecedented degree of precision~\cite{dowling2003quantum,nielsen2002quantum,white2020demonstration}. As such, the demand for a thorough analysis of the properties of the spin-boson model, as well as its many generalizations, is far from worn out.

The mathematical properties of the spin-boson model have been extensively analyzed in recent years; its spectrum has been investigated, and the existence and uniqueness of its ground state has been discussed~\cite{hirokawa2001remarks,hubner1995spectral,hirokawa1999expression,arai1990asymptotic,amann1991ground,davies1981symmetry,fannes1988equilibrium,hubner1995radiative,reker2020existence,hasler2021existence}. Going beyond the qubit case, a wider class of Hamiltonians describing the interaction between a quantum mechanical system and a structured boson field, known as \textit{generalized spin-boson} (GSB) models, was introduced by Arai and Hirokawa~\cite{arai1997existence}, and has been investigated as well~\cite{arai2000essential,arai2000ground,falconi2015self,takaesu2010generalized,teranishi2015self,teranishi2018absence}.

GSB models are defined as follows. Let $\mathfrak{h}$ be the Hilbert space describing the system, and $\focks(\hilb)$ the symmetric Fock space associated with a boson field, with $\hilb$ being its single-particle subspace. The free energy of the system and the field is associated with the following Hamiltonian on $\mathfrak{h}\otimes\focks(\hilb)$:
\begin{equation}
	H_0=A\otimes I+I\otimes\dOmega,
\end{equation}
with $A$ being the free Hamiltonian of the system, and $\dOmega$ being the free Hamiltonian of the boson field with dispersion relation $\omega$. GSB models are thus given by
\begin{equation}
	H_{f_1,\dots,f_r}=H_0+\lambda\sum_{j=1}^r\left(B_j\otimes\adag{f_j}+B_j^*\otimes\a{f_j}\right),
\end{equation}
with $\lambda\in\mathbb{R}$ being a coupling constant, $B_1,\dots,B_r$ a family of operators on the system, and $\a{f_j},\adag{f_j}$ being the creation and annihilation operators associated with a family of coupling functions $f_1,\dots,f_r$, which we denote as the form factors of the model.

GSB models include, among others, the spin-boson model on $\mathbb{C}^2\otimes\focks$:
\begin{equation}\label{eq:intro_sb}
H_f=H_0+\lambda\sigma_x\otimes\left(\a{f}+\adag{f}\right),
\end{equation}
a variant of the spin-boson model given by
\begin{equation}\label{eq:intro_rwasb}
	H_f=H_0+\lambda\left(\sigma_+\otimes\a{f}+\sigma_-\otimes\adag{f}\right),
\end{equation}
and a ``dephasing-type'' spin-boson model:
\begin{equation}\label{eq:intro_sb_pd}
	H_f=H_0+\lambda\sigma_z\otimes\left(\a{f}+\adag{f}\right),
\end{equation}
where, in Eqs.~\eqref{eq:intro_sb}--\eqref{eq:intro_rwasb},
\begin{equation}\label{eq:pauli}
\sigma_x=\begin{pmatrix}
0&1\\1&0
\end{pmatrix},\qquad\sigma_z=\begin{pmatrix}
1&0\\0&-1
\end{pmatrix},\qquad\sigma_+=\begin{pmatrix}
0&1\\0&0
\end{pmatrix},\qquad\sigma_-=\begin{pmatrix}
0&0\\1&0
\end{pmatrix},
\end{equation}
as well as their many-atom generalization; as we will discuss, the model in Eq.~\eqref{eq:intro_rwasb} can be obtained by neglecting counter-rotating terms in the spin-boson model~\eqref{eq:intro_sb}, a procedure often denoted as rotating-wave approximation (RWA)~\cite{agarwal1971rotating,agarwal1973rotating}. For a monochromatic boson field, the models in Eqs.~\eqref{eq:intro_sb} and~\eqref{eq:intro_rwasb} reduce to the well-known Rabi model~\cite{xie2017quantum,braak2011integrability,hwang2015quantum,zhong2013analytical} and Jaynes-Cummings model~\cite{shore1993jaynes,phoenix1991establishment,vogel1995nonlinear}, respectively.

GSB models (and, in particular, the models in Eqs.~\eqref{eq:intro_sb}--\eqref{eq:intro_sb_pd}) must obviously correspond to self-adjoint operators on the Hilbert space $\mathfrak{h}\otimes\focks(\hilb)$. A basic, and apparently natural,  assumption is the following one: the form factors $f_1,\dots,f_r$ must be \textit{normalizable}, that is, they must belong to the single-particle Hilbert space $\hilb$. From the mathematical point of view, this request ensures that the creation and annihilation operators $\adag{f}$, $\a{f}$ are closed operators on the Fock space: without such a condition, $\a{f}$ fails to be closed, so that its adjoint is not defined \cite{nelson1964interaction}.

Nevertheless, this assumption may be troublesome for applications: in formal calculations, physicists often make use of \textit{non-normalizable} form factors, e.g. Dirac distributions, thus dealing with operators whose very well-definiteness, not to mention self-adjointness, is questionable. Remarkably, non-normalizable form factors may come out from first principles; a basic example comes from waveguide quantum electrodynamics. The interaction between a single transverse mode of an electromagnetic field confined in an infinitely long waveguide and a pointlike quantum emitter can be described by a spin-boson model with $\hilb=L^2(\mathbb{R})$ and the following choices (in natural units) for the dispersion relation and the form factor:~\cite{facchi2019bound,lonigro2021stationary}
 \begin{equation}\label{eq:wqed}
 \omega(k)=\sqrt{k^2+m^2},\qquad	f(k)\propto\frac{\e^{-\ii kx_0}}{\sqrt[4]{k^2+m^2}},
 \end{equation}
 with $x_0$ being the position of the emitter in the guide, and $m>0$ an effective mass~\cite{jackson,dutra2005cavity}. Clearly, $f\notin\hilb$: the form factor does not decrease sufficiently quickly at $|k|\to\infty$. Furthermore, either as a byproduct of linear expansions or as \textit{a priori} toy models, one often encounters flat form factors $f(k)=\text{const.}$, modeling an idealized situation in which all field momenta are coupled to the spin with uniform strength, and corresponding, in the position representation, to Dirac distributions representing zero-range interactions. Interestingly, such choices of form factors have been shown to ensure the validity of quantum regression for specific classes of GSB models \cite{lonigro2022regression,lonigro2022beyond}.
 
 While such ``singular'' choices of form factors may be justified a posteriori via cutoff procedures or discretization arguments, a precise mathematical framework for non-normalizable form factors is desirable. This work represents a general effort in that direction: under minimal assumptions, we will define \textit{singular} GSB models which can accommodate form factors $f\notin\hilb$, provided that weaker constraints are fulfilled. Precisely, assuming $\omega(k)\geq m>0$ and denoting by $\hilb_{-s}$, $s\geq0$, the space of functions satisfying the condition
 \begin{equation}
\int\frac{|f(k)|^2}{\omega(k)^s}\,\mathrm{d}\mu(k)<\infty,
 \end{equation}
 the following results will be shown:
 \begin{itemize}
 	\item all generalized spin-boson models can accommodate, for small enough values of the coupling constant $\lambda$, form factors $f_1,\dots,f_r\in\hilb_{-1}$ (Prop.~\ref{prop:singgsb});
 	\item for the rotating-wave spin-boson model with form factor $f\in\hilb_{-1}$, the result can be improved in such a way to admit arbitrary values of $\lambda$, and by characterizing the operator domain and finding a closed expression for the resolvent (Theorem~\ref{thm:singrwa});
 	\item finally, via a renormalization procedure, a rotating-wave spin-boson model with form factor up to $f\in\hilb_{-2}$ can be defined, again with its operator domain and its resolvent being characterized (Theorem~\ref{thm:singrwa2}).
 \end{itemize}
In all cases listed above, the new models reduce to the ``regular'' ones when $f_1,\dots,f_r\in\hilb$; besides, it is always possible to approximate a singular GSB model by regular ones in the (either norm or strong) resolvent topology.
 
The main mathematical tool of our analysis will be the construction of scales of Hilbert spaces associated with self-adjoint operators: this will enable us to define $\a{f}$, $\adag{f}$, for all $f\in\hilb_{-s}$, as continuous maps between two properly chosen Hilbert spaces (Props.~\ref{prop:af_sing}--\ref{prop:adagf_sing}), instead that as unbounded operators on $\focks(\hilb)$, thus circumventing the aforementioned issue. Hilbert scales has been long applied to the study of singular perturbations of differential operators~\cite{albeverio2000singular,albeverio2007singularly,simon1995spectral,posilicano2001krein}, and were applied in order to introduce a singular Friedrichs-Lee Hamiltonian~\cite{derezinski2002renormalization,facchi2021spectral,lonigro2021selfenergy}, which indeed corresponds to the single-excitation sector of the model in Eq.~\eqref{eq:intro_rwasb}.

This work is organized as follows:
\begin{itemize}
	\item in Section~\ref{sec:fock} we sum up the basic definitions and properties of symmetric Fock spaces and operators on them, as well as the definition of generalized spin-boson (GSB) models with normalizable form factors;
	\item in Section~\ref{sec:singcreation} we introduce a scale of Fock spaces and we define creation and annihilation operators on the scale. These operators are compatible with the standard (regular) ones in the case of normalizable form factors, but can accommodate non-normalizable form factors;
	\item in Section~\ref{sec:singgsb} we introduce singular GSB models with form factors $f_1,\dots,f_r\in\hilb_{-1}$, proving their self-adjointness for small values of the coupling constant $\lambda$, and we show that every singular GSB model can be approximated by a sequence of regular GSB models;
	\item in Section~\ref{sec:singrwa} we study in greater detail the model in Eq.~\eqref{eq:intro_rwasb} with form factor $f\in\hilb_{-1}$. Improving the general results by following an alternative, nonperturbative strategy based on resolvent methods, we extend its structure to arbitrary values of $\lambda$, also characterizing its operator domain and computing its resolvent;
	\item in Section~\ref{sec:singrwa2} we improve the results of the previous section by showing that, via a renormalization procedure again based on resolvent methods, a further generalization of the model in Eq.~\eqref{eq:intro_rwasb} with form factor $f\in\hilb_{-2}$ (or, more specifically, $f\in\hilb_{-s}$ for $s\leq2$) can be obtained.
\end{itemize}
In the concluding section, further possible improvements of our results are discussed.\\

\textbf{Nomenclature}. We will denote by $\bar{z}$ the complex conjugate of a complex number $z\in\mathbb{C}$. Given a Borel measure space $(X,\mu)$, the Lebesgue integral of a measurable function on $X$ will be denoted by
\begin{equation}
\int f(k)\,\mathrm{d}\mu(k)\quad\text{or}\quad	\int f(k)\,\mathrm{d}\mu,
\end{equation}
the second expression being used whenever there is no risk of confusion. An analogous notation will be used for multiple integrals:
\begin{equation}
\int f(k_1,\dots,k_n)\:\mathrm{d}^n\mu.
\end{equation}
Given a Hilbert space $\mathcal{K}$, the scalar product on $\mathcal{K}$ and its associated norm will be denoted by
\begin{equation}
\Braket{\Psi,\Phi}_{\mathcal{K}},\qquad \|\Psi\|_{\mathcal{K}}=\Bigl(\Braket{\Psi,\Psi}_{\mathcal{K}}\Bigr)^{1/2};
\end{equation}
the scalar product is linear at the right and antilinear at the left. In particular, given a Borel measure space $(X,\mu)$, we will denote by $L^2_\mu(X)$ the space of square-integrable functions on $X$ endowed with the scalar product
\begin{equation}
\Braket{f,g}_{L^2_\mu(X)}=\int\overline{f(x)}g(x)\;\mathrm{d}\mu(x).
\end{equation}
Given a (possibly unbounded) closed linear operator $T$ on $\mathcal{K}$, its domain and (if applicable) form domain shall be denoted via $\mathcal{D}(T)$ and $\mathcal{Q}(T)$; we shall denote by $T^*$ the adjoint of $T$, defined via
\begin{equation}
\Braket{\Psi,T\Phi}_{\mathcal{K}}=\Braket{T^*\Psi,\Phi}_{\mathcal{K}},\qquad \Phi\in\mathcal{D}(T),\;\Psi\in\mathcal{D}(T^*).
\end{equation}
Bounded (or, equivalently, continuous) operators are understood to be defined with domain $\mathcal{D}(T)=\mathcal{K}$. The space of bounded operators on $\mathcal{K}$ will be denoted as $\mathcal{B}(\mathcal{K})$. Finally, given two distinct Hilbert spaces $\mathcal{K}_1,\mathcal{K}_2$, the space of bounded (continuous) operators between $\mathcal{K}_1$ and $\mathcal{K}_2$ will be denoted by $\mathcal{B}(\mathcal{K}_1,\mathcal{K}_2)$.

\section{Preliminaries: operators on Fock spaces}\label{sec:fock}
For completeness, and to fix the notation, we will recall in the present section some known properties of Fock spaces that will be needed in our discussion. Subsection~\ref{subsec:fock} is devoted to the basic definitions, while Subsection~\ref{subsec:creation} is devoted to the (regular) bosonic creation and annihilation operators. See e.g.~\cite{reed1975fourier,derezinski2013mathematics,bratteli1987operator,folland2021quantum} for a thorough introduction to the subject. Finally, in Subsection~\ref{subsec:gsb} we introduce the class of  (regular) generalized spin-boson models (GSB) and discuss some examples.

\subsection{Fock spaces and second-quantized operators}\label{subsec:fock}
Here we will recall the main properties of symmetric Fock spaces and discuss some operators on them. For simplicity, we shall always consider Fock spaces constructed on a Hilbert space $\hilb=L^2_\mu(X)$, with $(X,\mu)$ being a Borel measure space; however, our discussion will be largely independent of this choice. 
\begin{definition}[Symmetric Fock space]
	Let $\hilb^{(0)}=\mathbb{C}$ and, for $n\geq1$,
	\begin{equation}
	\hilbn=\bigotimes_{j=1}^n\hilb\simeq\left\{\Psi^{(n)}:X^n\rightarrow\mathbb{C}:\int|\Psi^{(n)}(k_1,\dots,k_n)|^2\,\mathrm{d}^n\mu<\infty\right\},
	\end{equation}
endowed with the scalar product
\begin{equation}
	\Braket{\Psi^{(n)},\Phi^{(n)}}_{\hilbn}=\int\overline{\Psi^{(n)}(k_1,\dots,k_n)}\Phi^{(n)}(k_1,\dots,k_n)\,\mathrm{d}^n\mu.
\end{equation}
The \textit{symmetric Fock space} $\fock(\hilb)$ on $\hilb$ is the space
\begin{equation}
\fock(\hilb)=\bigoplus_{n\in\mathbb{N}}S_n\hilbn
\end{equation}
with $S_n$ being the symmetrization operator on $\hilbn$, i.e.
\begin{equation}
	\left(S_n\Psi^{(n)}\right)(k_1,\ldots,k_n)=\frac{1}{n!}\sum_{\sigma\in\mathfrak{S}_n}\Psi^{(n)}(k_{\sigma_1},\ldots,k_{\sigma_n}),
\end{equation}
with $\mathfrak{S}_n$ being the group of permutations on $\{1,\dots,n\}$. The space $\hilb^{(n)}$ will be referred to as the $n$-particle subspace of $\fock(\hilb)$; the vacuum state of the field, i.e. the unique (up to a phase) normalized element of $\hilb^{(0)}$, will be denoted as $\Omega$. 
\end{definition}
In the following we will use the shorthand $\fock(\hilb)\equiv\fock$. The scalar product and norm on the symmetric Fock space are therefore given by
\begin{equation}
\Braket{\Psi,\Phi}_{\fock}=\sum_{n\in\mathbb{N}}\braket{\Psi^{(n)},\Phi^{(n)}}_{\hilbn},\quad\bigl\|\Psi\bigr\|^2_{\fock}=\sum_{n\in\mathbb{N}}\left\|\Psi^{(n)}\right\|^2_{\hilbn};
\end{equation}
the elements of $\fock$ are sequences $\Psi=\{\Psi^{(n)}\}_{n\in\mathbb{N}}$ such that $\|\Psi\|_\fock<\infty$ and, for all $n\in\mathbb{N}$, $S_n\Psi^{(n)}=\Psi^{(n)}$, that is, they are invariant under any permutation of the variables:
\begin{equation}
	\Psi^{(n)}(k_1,\ldots,k_n)=\Psi^{(n)}(k_{\sigma_1},\ldots,k_{\sigma_n})
\end{equation}
for all $\sigma\in\mathfrak{S}_n$. Necessarily, all operators on $\focks$ must map completely symmetric states into completely symmetric states.

\begin{definition}[Second quantization]
Let $T$ be a densely defined, closed operator on the single-particle space $\hilb$. Its \textit{second quantization} $\mathrm{d}\Gamma(T)$ is the operator on $\focks$ defined via
\begin{equation}
	\mathrm{d}\Gamma(T)=\bigoplus_{n\in\mathbb{N}} T^{(n)},\qquad T^{(n)}=\sum_{j=1}^n\left( I\otimes\dots\otimes\overbrace{T}^{j\text{th}}\otimes\dots\otimes I\right),
\end{equation}
where we set $T^{(0)}=0$.
\end{definition}
By definition, its domain is given by
\begin{equation}
\mathcal{D}\left(\mathrm{d}\Gamma(T)\right)=\left\{\Psi\in\focks:\,\Psi^{(n)}\in\mathcal{D}\left(T^{(n)}\right),\;\sum_{n\in\mathbb{N}}\left\|T^{(n)}\Psi^{(n)}\right\|_{\hilbn}^2<\infty\right\},
\end{equation}
where, for all $n\geq1$,
\begin{equation}
\mathcal{D}\left(T^{(n)}\right)=\bigotimes_{j=1}^n\mathcal{D}(T).
\end{equation}
Notice that $\mathrm{d}\Gamma(T)$ is a legitimate operator on the symmetric Fock space $\focks$ since, by construction, it preserves the complete symmetry of the vectors. By the properties of direct sums (see e.g.~\cite{nussbaum1964reduction}), it is a densely defined closed operator on $\focks$, and it is self-adjoint if and only if $T$ is self-adjoint. Two fundamental examples follow. \\

\textbf{Number operator}. The second quantization of the identity on $\hilb$, $N=\mathrm{d}\Gamma(I)$, is the number operator:
\begin{equation}
\mathcal{D}(N)=\left\{\Psi\in\focks:\,\sum_{n\in\mathbb{N}} n^2\left\|\Phi^{(n)}\right\|^2_{\hilbn}<\infty\right\},\quad N\Psi^{(n)}=n\Psi^{(n)},
\end{equation}
its spectrum being the set $\mathbb{N}$ of nonnegative integers. Notice that, while $I$ is a bounded operator on $\hilb$, $N$ is obviously unbounded on $\focks$; in general, the second quantization operator of every single-particle operator but the null one is an unbounded operator. Physically, since the number of particles is allowed to be arbitrarily large, so is the average value of every single-particle observable.\\

\textbf{Second quantization of a multiplication operator}. Let $\omega$ a real-valued Borel measurable function; with a slight abuse of notation, the same symbol $\omega$ will be used for the multiplication operator associated with it, that is,
\begin{equation}
\left(\omega\psi\right)(k)=\omega(k)\psi(k),
\end{equation}
with domain
\begin{equation}
\mathcal{D}(\omega)=\left\{\psi\in\hilb:\,\int\omega(k)^2\,|\psi(k)|^2\,\mathrm{d}\mu<\infty \right\};
\end{equation}
then the operator $\omega^{(n)}$ on $\hilbn$ acts, for $n\geq1$, as
\begin{equation}
\left(\omega^{(n)}\Psi^{(n)}\right)(k_1,\dots,k_n)=\left(\sum_{j=1}^n\omega(k_j)\right)\Psi^{(n)}(k_1,\dots,k_n),
\end{equation}
and the second quantization of $\omega$ on the symmetric Fock space has domain
\begin{equation}
\mathcal{D}(\dOmega)=\left\{\Psi\in\focks:\,\sum_{n\in\mathbb{N}}\int\left(\sum_{j=1}^n\omega(k_j)\right)^2\left|\Psi^{(n)}(k_1,\dots,k_n)\right|^2\;\mathrm{d}^n\mu<\infty\right\}.
\end{equation}
The following straightforward property holds:
\begin{proposition}\label{prop:numb}
	Suppose that
	\begin{equation}
	m=\inf_{k\in X}\omega(k)\geq0;
	\end{equation}
	then, for all $s\geq0$, we have $\mathcal{D}(\dOmega^{s/2})\subset\mathcal{D}(N^{s/2})$ and, for all $\Psi\in\mathcal{D}(\dOmega^{s/2})$,
	\begin{equation}\label{eq:diseq}
	\left\|\dOmega^{s/2}\Psi\right\|_\fock\geq m^{s/2}\,\left\|N^{s/2}\Psi\right\|_\fock.
	\end{equation}
\end{proposition}
\begin{proof}
Trivial consequence of the inequality $\omega(k)\geq m\geq0$.
\end{proof}
For simplicity, we shall always suppose $m>0$ hereafter. We remark that, even if $m>0$, $\dOmega$ is \textit{not} a strictly positive operator, since $\dOmega\Omega=0$; however, $\dOmega$ is indeed strictly positive on all sectors $\hilb^{(n)}$ with $n\geq1$.

\subsection{Creation and annihilation operators}\label{subsec:creation}
We will now introduce the bosonic creation and annihilation operators on $\focks$ associated with an element $f\in\hilb$ of the single-particle space $\hilb$, which enter crucially in the definition of generalized spin-boson models. The primary goal of Section~\ref{sec:singcreation} will be to generalize the construction presented here.

\begin{definition}[Creation and annihilation operators]
	Let $f\in\hilb$. The \textit{creation operator} $a^\dag(f)$ and the \textit{annihilation operator} $a(f)$ are the operators on $\focks$ with domain\footnote{Usually (see e.g.~\cite{reed1975fourier,bratteli1987operator}), the creation and annihilation operators are equivalently introduced by defining them on the (dense) subspace of all Fock states with finite number of particles, that is, $\Psi^{(n)}=0$ for sufficiently large $n$, and then taking the closure.} $\mathcal{D}(\a{f})=\mathcal{D}(\adag{f})$ given by
	\begin{equation}
\mathcal{D}(\a{f})=\left\{\Psi\in\focks:\,\sum_{n\in\mathbb{N}}n \left|\int\overline{f(k_n)}\Psi^{(n)}(k_1,\dots,k_{n-1},k_n)\;\mathrm{d}\mu(k_n)\right|^2<\infty\right\}
	\end{equation}	
	acting as follows: $\a{f}\Omega=0$ and, for $n\geq1$,
	\begin{equation}\label{eq:af_explicit}
		\left(\a{f}\Psi^{(n)}\right)(k_1,\dots,k_{n-1})=\sqrt{n}\int\overline{f(k_n)}\Psi^{(n)}(k_1,\dots,k_{n-1},k_n)\;\mathrm{d}\mu(k_n),
	\end{equation}
	and, for all $n\geq0$,
	\begin{eqnarray}\label{eq:adagf_explicit}
		\left(\adag{f}\Psi^{(n)}\right)(k_1,\dots,k_{n},k_{n+1})&=&\frac{1}{\sqrt{n+1}}\bigg(\sum_{j=1}^n\Psi^{(n)}(k_1,\dots,\overbrace{k_{n+1}}^{j\text{th}},\dots,k_n)f(k_j)\nonumber\\&&+\Psi^{(n)}(k_1,\dots,k_n)f(k_{n+1})\bigg).
	\end{eqnarray}
\end{definition}
By construction, for all $n\in\mathbb{N}$,
\begin{equation}\label{eq:ladder}
	\a{f}\hilb^{(n+1)}\subset\hilb^{(n)},\qquad\adag{f}\hilb^{(n)}\subset\hilb^{(n+1)}.
\end{equation} 
Both $a(f),a^\dag(f)$ are known to be densely defined, closed and unbounded operators on $\focks$ satisfying
\begin{equation}
	\a{f}^*=\adag{f},
\end{equation}
i.e. they are mutually adjoint~\cite{reed1975fourier,nelson1964interaction,falconi2015self}, and satisfy the well-known commutation property:
\begin{equation}\label{eq:commute}
\left[\a{f},\adag{g}\right]\Psi:=\left(\a{f}\adag{g}-\adag{g}\a{f}\right)\Psi=\braket{f,g}\,\Psi
\end{equation}
for all vectors $\Psi\in\focks$ such that the left-hand side of Eq.~\eqref{eq:commute} is well-defined. Physically, the unboundedness of the creation operators reflects the absence of a bound of the number of bosons which can occupy a given state, differently from what happens to fermions~\cite{bratteli1987operator}.

In particular, their natural domain contains $\mathcal{D}(N^{1/2})$, i.e. the space of all Fock states with a finite average number of particles:
\begin{proposition}\label{prop:domaf}
	Let $f\in\hilb$; then $\mathcal{D}(\a{f})\supset\mathcal{D}(N^{1/2})$.
\end{proposition}
\begin{proof}
Let $\Psi\in\mathcal{D}(N^{1/2})$. Then, applying the Cauchy-Schwartz inequality, for all $n\geq1$ we have
\begin{eqnarray}
	\left|\left(\a{f}\Psi^{(n)}\right)(k_1,\dots,k_{n-1})\right|^2&=&\left|\sqrt{n}\int\overline{f(k_n)}\Psi^{(n)}(k_1,\dots,k_{n-1},k_n)\,\mathrm{d}\mu(k_n)\right|^2\nonumber\\
	&\leq&n\|f\|^2\int|\Psi^{(n)}(k_1,\dots,k_{n-1},k_n)|^2\,\mathrm{d}\mu(k_n),
\end{eqnarray}
hence, integrating on the variables $k_1,\dots,k_{n-1}$,
\begin{eqnarray}
\left\|\a{f}\Psi^{(n)}\right\|_{\hilb^{(n-1)}}^2\leq\|f\|^2\,n\|\Psi^{(n)}\|^2_{\hilbn}
\end{eqnarray}
and therefore
\begin{equation}
\left\|\a{f}\Psi\right\|^2_\fock\leq\|f\|^2\sum_{n\in\mathbb{N}}n\|\Psi^{(n)}\|^2_{\hilbn}=\|f\|^2\|N^{1/2}\Psi\|^2_\fock<\infty,
\end{equation}
implying $\Psi\in\mathcal{D}(\a{f})$.
\end{proof}
As an immediate consequence of Props.~\ref{prop:numb} and~\ref{prop:domaf}, for every $f\in\hilb$ we have
\begin{equation}\label{eq:domaf}
	\mathcal{D}(a(f))\supset\mathcal{D}(\dOmega^{1/2})\supset\mathcal{D}(\dOmega).
\end{equation}

\subsection{Generalized spin-boson (GSB) models}\label{subsec:gsb}
We can now introduce the class of generalized spin-boson (GSB) models.
\begin{definition}
	Let $\mathfrak{h}$ a Hilbert space and $A\in\mathcal{B}(\mathfrak{h})$ a nonnegative bounded\footnote{Here, to keep the discussion simple, we are only considering bounded operators on $\hfrak$ (which is indeed the case whenever $\hfrak$ is finite-dimensional, as it usually is in applications); however, GSB models with unbounded operators on the space $\mathfrak{h}$ could indeed be considered, see e.g.~\cite{arai1997existence}. The nonnegativity hypothesis is also easily amendable.} self-adjoint operator; define the self-adjoint operator on the Hilbert space $\hfrak=\mathfrak{h}\otimes\focks$ via
	\begin{equation}\label{eq:def_h0}
		H_0=A\otimes I+I\otimes\dOmega.
	\end{equation}
	with domain $\mathcal{D}(H_0)=\mathfrak{h}\otimes\mathcal{D}(\dOmega)$.
	
	Given $f_1,\dots,f_r\in\hilb$,  $B_1,\dots,B_r\in\mathcal{B}(\mathfrak{h})$, and a coupling constant $\lambda\in\mathbb{R}$, a \textit{generalized spin-boson model} (GSB model)~\cite{arai1997existence} is an operator on $\hfrak$, with domain $\mathcal{D}(H_{f_1,\dots,f_r})=\mathcal{D}(H_0)$, defined via
	\begin{equation}\label{eq:def_gsb}
		H_{f_1,\dots,f_r}=H_0+\lambda\sum_{j=1}^r\left(B_j\otimes a^\dag(f_j)+B_j^*\otimes a(f_j)\right)
	\end{equation}
	 $f_1,\dots,f_r$ are the \textit{form factors} of the model. 
\end{definition}
Notice that the operator in Eq.~\eqref{eq:def_gsb} is well-defined on $\mathcal{D}(H_0)$ because of Eq.~\eqref{eq:domaf}. The physical meaning of this model is transparent: $H_0$ is the Hamiltonian associated with the free energy of a quantum system (e.g. an ensemble of atoms), with free Hamiltonian $A$, and a boson field with dispersion relation $\omega$; the interaction term is constructed in such a way that either
\begin{itemize}
	\item a boson with wavefunction $f_j$ is annihilated and the operator $B_j^*$ is applied to the system, or
	\item a boson with wavefunction $f_j$ is created and the operator $B_j$ is applied to the system.
\end{itemize}
For every choice of the parameters, GSB models can be shown to be self-adjoint operators~\cite[Prop.~1.1]{arai1997existence}.
\begin{remark}
An alternative (and equivalent) representation of such models, which is often found in the literature, is the following one:
	\begin{equation}\label{eq:def_gsb2}
H_{f_1,\dots,f_r}=H_0+\lambda\sum_{j=1}^rC_j\otimes\phi(f_j),\qquad C_j=C_j^*,
\end{equation}
where $\phi(f)$, denoted as the Segal field operator~\cite{reed1975fourier}, is simply defined via
\begin{equation}
	\phi(f)=\frac{1}{\sqrt{2}}\left(a(f)+a^\dag(f)\right).
\end{equation}
\end{remark} 
Reprising the discussion in the introductory section, let us elaborate more on the particular GSB models listed in the Introduction, cf.~Eqs.~\eqref{eq:intro_sb}--\eqref{eq:intro_sb_pd}.\\

\textbf{The spin-boson model}. Given $\mathfrak{h}=\mathbb{C}^2$, let $H_0$ be as in Eq.~\eqref{eq:def_h0} where we set
\begin{equation}\label{eq:freesb}
A=\begin{pmatrix}
\omega_{\mathrm{e}}&0\\0&\omega_{\mathrm{g}}
\end{pmatrix},
\end{equation}
the latter being the energy of a two-level system with excited and ground energy respectively equal to $\omega_{\mathrm{e}}$ and $\omega_{\mathrm{g}}$. The \textit{spin-boson model} is defined via
\begin{equation}\label{eq:def_sb}
H_{f}=H_0+\lambda\sigma_x\otimes\left(a(f)+a^\dag(f)\right),
\end{equation}
with $\sigma_x$ being the first Pauli matrix, cf.~Eq.~\eqref{eq:pauli}. This model describes the interaction between a two-level system (spin) and a structured boson field, and is encountered in many branches of physics (see references in the Introduction). We remark that, when choosing the boson field to be monochromatic (i.e. $\mu$ is a Dirac measure), the spin-boson model reduces to the Rabi model. A generalization of the model, describing an ensemble of $r$ atoms each solely interacting with the field, can be easily constructed by choosing $\mathfrak{h}=\mathbb{C}^{2r}$,
\begin{equation}\label{eq:freesbn}
	A=\bigoplus_{j=1}^r\left( I\otimes\dots\otimes A_j\otimes\dots\otimes I\right),\qquad A_j=\begin{pmatrix}
		\omega^j_{\mathrm{e}}&0\\0&\omega^j_{\mathrm{g}}
	\end{pmatrix},
\end{equation}
and
\begin{equation}
	H_f=H_0+\lambda\sum_{j=1}^r\sigma^j_x\otimes\left(\a{f_j}+\adag{f_j}\right),
\end{equation}
with $H_0$ again as in Eq.~\eqref{eq:h02}, and
\begin{equation}
	\sigma_x^j=\bigoplus_{j=1}^r\left( I\otimes\dots\otimes \overbrace{\sigma_x}^{j\text{th}}\otimes\dots\otimes I\right);
\end{equation}
in the monochromatic case, this is called the Dicke model~\cite{garraway2011dicke,emary2003chaos,wang1973phase}.\\

\textbf{The rotating-wave spin-boson model}. A variation of the spin-boson model introduced above is the following:
\begin{equation}\label{eq:def_sbrwa}
H_{f}=H_0+\lambda\left(\sigma_+\otimes a(f)+\sigma_-\otimes a^\dag(f)\right),
\end{equation}
with the matrices $\sigma_\pm$ as in Eq.~\eqref{eq:pauli}. Since $\sigma_z=\sigma_+-\sigma_-$, it is immediate to show that the Hamiltonian above can be obtained by expanding the spin-boson model in Eq.~\eqref{eq:def_sb} and neglecting the two terms $\sigma_+\otimes a^\dag(f)$ and $\sigma_-\otimes a(f)$. Such a procedure is usually referred to as a \textit{rotating-wave approximation} (RWA), which is often invoked in the small-coupling regime. As we will extensively discuss in Section~\ref{sec:singrwa}, this model preserves the total number of excitations (Prop.~\ref{prop:numexc}), which makes it far easier to solve. The model is also fundamental in the theory of open quantum systems: its reduced dynamics on $\mathbb{C}^2$ corresponds to the amplitude-damping (AD) qubit channel.

We will refer to this model as the rotating-wave spin-boson model; this is sometimes also referred to as a Wigner-Weisskopf model~\cite{hirokawa2001remarks}. When choosing a monochromatic field, it reduces to the \textit{Jaynes-Cummings model}, which is again ubiquitous in quantum optics.

Finally, an $r$-atom generalization of this model can be readily constructed:
\begin{equation}
	H_f=H_0+\lambda\sum_{j=1}^r\left(\sigma^j_+\otimes\a{f_j}+\sigma^j_-\otimes\adag{f_j}\right),
\end{equation}
where
\begin{equation}
	\sigma_\pm^j=\bigoplus_{j=1}^r\left( I\otimes\dots\otimes \overbrace{\sigma_\pm}^{j\text{th}}\otimes\dots\otimes I\right);
\end{equation}
in the monochromatic case, this is known as the Tavis-Cummings model~\cite{bogoliubov1996exact,fink2009dressed}.\\

\textbf{The dephasing-type spin-boson model}. Another interesting model, which we shall refer to as the \textit{dephasing-type spin-boson model}, is the following one:
\begin{equation}\label{eq:def_sb_pd}
	H_{f}=H_0+\lambda\sigma_z\otimes\left(a(f)+a^\dag(f)\right),
\end{equation}
which is obtained by replacing $\sigma_x$, in Eq.~\eqref{eq:def_sb}, with the third Pauli matrix $\sigma_z$, cf.~Eq.~\eqref{eq:pauli}. 

Despite its similarity with the spin-boson models in Eq.~\eqref{eq:def_sb}--\eqref{eq:def_sbrwa}, the dynamics induced by the Hamiltonian in Eq.~\eqref{eq:def_sb_pd} greatly differs from those since it involves no transitions between the two spin states: in fact, while the reduced dynamics on $\mathbb{C}^2$ induced by the model in Eq.~\eqref{eq:def_sbrwa} is an amplitude-damping (AD) qubit channel, the one induced by Eq.~\eqref{eq:def_sb_pd} is a phase-damping (PD), or dephasing, qubit channel. As such, the dephasing-type spin-boson model represents the paradigmatic toy model to study decoherence phenomena in open quantum systems.

\section{Singular creation and annihilation operators}\label{sec:singcreation}

The first step towards a rigorous implementation of GSB models with non-normalizable form factors necessarily involves a redefinition of the creation and annihilation operators $\adag{f},a(f)$ introduced in the previous section. The basic idea is simple: even if
\begin{equation}\label{eq:normalization}
	\int|f(k)|^2\,\mathrm{d}\mu=+\infty,
\end{equation}
we may still have\footnote{It is worth mentioning that, when $m=0$, this is no longer the case: Eq.~\eqref{eq:normalization} is not a stronger condition than Eq.~\eqref{eq:weaker}. In fact, mathematicians have often analyzed the \textit{converse} situation: $f$ is normalizable but Eq.~\eqref{eq:weaker} does not hold. In physicists' jargon, this is an example of an infrared divergence, while, in the present paper, we are rather analyzing form factors with an ultraviolet divergence.}
\begin{equation}\label{eq:weaker}
	\int\frac{|f(k)|^2}{\omega(k)^s}\,\mathrm{d}\mu<+\infty
\end{equation}
for some $s>0$, since we are assuming $\omega(k)\geq m>0$. An example is the form factor in Eq.~\eqref{eq:wqed} in the Introduction, which is not normalizable but satisfies the weaker normalization constraint~\eqref{eq:weaker} for every $s>0$.

This simple observation will lead us to the concept of scales of Hilbert spaces. After recalling in Subsection~\ref{subsec:scale} the standard construction of the scale of Hilbert spaces associated with a nonnegative self-adjoint operator, in Subsection~\ref{subsec:creationscalefock} we will introduce two scales associated with the operators $\omega$ and $\dOmega$, and define a generalization of creation and annihilation operators as continuous maps on them, which will allow us to give a sense to the formal expressions $\a{f},\adag{f}$ even if $f$ is not normalizable, provided that Eq.~\eqref{eq:weaker} holds for some $s\geq1$. The approximation of singular creation and annihilation operators via sequences of standard ones will be finally discussed in Subsection~\ref{subsec:approx}.

This construction will be crucially employed in Sections~\ref{sec:singgsb}--\ref{sec:singrwa2} in order to construct GSB models with non-normalizable form factors.

\subsection{Generalities on scales of Hilbert spaces}\label{subsec:scale}
We will now revise here, in an abstract setting, some basic definition and properties of the scales of Hilbert spaces; see e.g.~\cite{albeverio2000singular,albeverio2007singularly,simon1995spectral} for further details and applications. The reader familiar with this formalism may jump directly to Subsection~\ref{subsec:creationscalefock}.
\begin{definition}
	Let $\mathcal{K}$ be a Hilbert space and $T_0$ a nonnegative self-adjoint operator on it. For all $s\in\mathbb{R}$, the space $\mathcal{K}_s$ is the completion of the set $\bigcap_{n\in\mathbb{N}}\mathcal{D}(T_0^n)$ with respect to the scalar product\footnote{We may equivalently use $|T_0-z_0|$ for any $z_0$ in the resolvent set of $T_0$; all such choices would yield equivalent norms. In particular, if $T_0$ is strictly positive (i.e. $T_0\geq\epsilon>0$ for some $\epsilon>0$), we can replace $T_0+1$ with $T_0$ in Eq.~\eqref{eq:scalar} as well as in the remainder of this discussion.}
	\begin{equation}\label{eq:scalar}
		\phi,\psi\mapsto\braket{\phi,\psi}_{\mathcal{K}_s}=\Braket{(T_0+1)^{s/2}\phi,(T_0+1)^{s/2}\xi}_{\mathcal{K}}.
	\end{equation}
	The family of spaces $\{\mathcal{K}_s\}_{s\in\mathbb{R}}$ is the \textit{scale of Hilbert spaces} associated with $T_0$ ($T_0$-scale).
\end{definition}
The denomination follows from the fact that, by construction, we have $\mathcal{K}_s\subset\mathcal{K}_{s'}$ whenever $s\geq s'$, all inclusions being dense. In particular, $\mathcal{K}_0=\mathcal{K}$, while $\mathcal{K}_1$ and $\mathcal{K}_2$ coincide respectively with the form domain $\mathcal{Q}(T_0)$ and the domain $\mathcal{D}(T_0)$ of the operator $T_0$.

By construction, for all $r,s\in\mathbb{R}$ the operator $(T_0+1)^{r/2}$ can be continuously extended to a continuous isometry between the Hilbert spaces $\mathcal{K}_s$ and $\mathcal{K}_{s-r}$; with an abuse of notation, we will still denote such an operator with the same symbol and say, in the various cases, that $T_0$ is ``interpreted'' either as an unbounded operator on $\mathcal{K}$, or as a bounded operator between two members of the scale.

Besides, for all $s\geq0$, the spaces $\mathcal{K}_{+s}$ and $\mathcal{K}_{-s}$ are dual under the pairing
\begin{equation}\label{eq:dualitypairing}
	\phi\in\mathcal{K}_{+s},\psi\in\mathcal{K}_{-s}\mapsto\left(\phi,\psi\right)_{\mathcal{K}_{+s},\mathcal{K}_{-s}}=\Braket{(T_0+1)^{s/2}\phi,(T_0+1)^{-s/2}\psi}_{\mathcal{K}},
\end{equation}
with $(T_0+1)^{\pm s/2}$ above being interpreted as isometries between $\mathcal{K}_{\pm s}$ and $\mathcal{K}$. We note that the triple of Hilbert spaces $(\mathcal{K}_{-s},\mathcal{K},\mathcal{K}_{+s})$ is an example of a Gelfand triple~\cite{bohm1974rigged,de2005role,bohm1989dirac}. Eq.~\eqref{eq:dualitypairing} easily implies a Cauchy-Schwartz-like inequality:
\begin{equation}\label{eq:cauchyschwartzlike}
	\left|\left(\phi,\psi\right)_{\mathcal{K}_{+s},\mathcal{K}_{-s}}\right|\leq\|\phi\|_{\mathcal{K}_{+s}}\|\psi\|_{\mathcal{K}_{-s}}.
\end{equation}
Continuous operators between $\mathcal{K}_{+s}$ and $\mathcal{K}_{-s}$ are also associated with sesquilinear forms. Indeed, given $T_1\in\mathcal{B}(\mathcal{K}_{+s},\mathcal{K}_{+s})$, we can define the form
\begin{equation}
	\phi,\xi\in\mathcal{K}_{+s}\mapsto t_1(\phi,\xi)=\left(\phi,T_1\xi\right)_{\mathcal{K}_{+s},\mathcal{K}_{-s}}\in\mathbb{C}.
\end{equation}
This can be interpreted as an unbounded form on $\mathcal{K}$ with domain $\mathcal{K}_{+s}$; we will say that $T_1$ is symmetric if the associated form is symmetric. In particular, $T_0$ itself, interpreted as an operator between $\mathcal{K}_{+1}$ and $\mathcal{K}_{-1}$, is uniquely associated with the sesquilinear form
\begin{equation}
	\phi,\xi\in\mathcal{K}_{+1}\mapsto t_0(\phi,\xi)=\left(\phi,T_0\xi\right)_{\mathcal{K}_{+1},\mathcal{K}_{-1}}\in\mathbb{C}.
\end{equation}
As a simple consequence of the well-known KLMN theorem~\cite{teschl2009mathematical,simon2015quantum}, the correspondence between sesquilinear forms and continuous operators between $\mathcal{K}_{+1}$ and $\mathcal{K}_{-1}$ can be used to define unbounded self-adjoint operators on $\mathcal{K}$ by means of continuous operators on the Fock scale. For future convenience, we will state the result explicitly.
\begin{proposition}\label{prop:klmn_revisited}
	Let $T_0$ a nonnegative self-adjoint operator on $\mathcal{K}$, $\{\mathcal{K}_s\}_{s\in\mathbb{R}}$ the associated $T_0$-scale, and $T_1\in\mathcal{B}(\mathcal{K}_{+1},\mathcal{K}_{-1})$ symmetric. Then, for sufficiently small $\lambda$, the continuous operator
	\begin{equation}
T_\lambda\equiv T_0+\lambda T_1\in\mathcal{B}\left(\mathcal{K}_{+1},\mathcal{K}_{-1}\right)
	\end{equation}
	is uniquely associated with a self-adjoint operator on $\mathcal{K}$ with form domain equal to $\mathcal{Q}(T_0)=\mathcal{K}_{+1}$.
\end{proposition}
\begin{proof}
	The sesquilinear form defined via\begin{equation}
		t_1(\phi,\xi)=\left(\phi,T_1\xi\right)_{\mathcal{K}_{+1},\mathcal{K}_{-1}}
	\end{equation}
	satisfies, for all $\phi\in\mathcal{K}_{+1}$, the inequality
	\begin{eqnarray}
		|t_1(\phi,\phi)|&\leq&\|\phi\|_{\mathcal{K}_{+1}}\left\|T_1\phi\right\|_{\mathcal{K}_{-1}}\nonumber\\
		&\leq&\|T_1\|_{\mathcal{B}(\mathcal{K}_{+1},\mathcal{K}_{-1})}\,\|\phi\|^2_{\mathcal{K}_{+1}}\nonumber\\
		&=&\|T_1\|_{\mathcal{B}(\mathcal{K}_{+1},\mathcal{K}_{-1})}\,\|(T_0+1)^{1/2}\phi\|^2_{\mathcal{K}}\nonumber\\
		&=&\|T_1\|_{\mathcal{B}(\mathcal{K}_{+1},\mathcal{K}_{-1})}\left(t_0(\phi,\phi)+\|\phi\|_{\mathcal{K}}^2\right)
	\end{eqnarray}
	as an immediate consequence of the Cauchy-Schwartz inequality~\eqref{eq:cauchyschwartzlike} and the boundedness of $T_1$ as an operator between $\mathcal{K}_{+1}$ and $\mathcal{K}_{-1}$. Therefore, the sesquilinear form $t_1(\cdot,\cdot)$ is relatively bounded with respect to $t_0(\cdot,\cdot)$. By the KLMN theorem, this implies that, whenever its relatively bound is less than one, and thus
	\begin{equation}
		\lambda<\frac{1}{\|T_1\|_{\mathcal{B}(\mathcal{K}_{+1},\mathcal{K}_{-1})}},
	\end{equation}
	the sesquilinear form associated with $T_0+\lambda T_1$ is uniquely associated with a self-adjoint operator with form domain equal to $\mathcal{Q}(T_0)=\mathcal{K}_{+1}$.
\end{proof}
\begin{remark}\label{remark:klmn}
Prop.~\ref{prop:klmn_revisited} can be equivalently stated as follows: for sufficiently small $\lambda$, there is a dense subspace $\mathcal{D}(T_\lambda)\subset\mathcal{K}$ such that the restriction of $T_\lambda\in\mathcal{B}(\mathcal{K}_{+1},\mathcal{K}_{-1})$ to $\mathcal{D}(T_\lambda)$ defines an unbounded self-adjoint operator on the Hilbert space $\mathcal{K}$, which (with the usual abuse of notation) we still denote by $T_\lambda$. 

We remark that Prop.~\ref{prop:klmn_revisited}, while ensuring the existence of such a domain, does not provide additional information about it; in general, $\mathcal{D}(T_\lambda)\neq\mathcal{D}(T_0)$ depends nontrivially on the coupling constant $\lambda$.
\end{remark}

\subsection{Creation and annihilation operators on the $\dOmega$-scale}\label{subsec:creationscalefock}
Coming back to our original problem, we will now introduce two important scales of Hilbert spaces:
\begin{itemize}
	\item the $\omega$-scale $\{\hilb_s\}_{s\in\mathbb{R}}$ associated with $\omega$, with
	\begin{equation}
		\|\psi\|^2_{\hilb_s}=\|\omega^{s/2}\psi\|^2_\hilb=\int\omega(k)^s|\psi(k)|^2\,\mathrm{d}\mu;
	\end{equation}
	for brevity, we shall set $\|\psi\|_{\hilb_s}\equiv\|\psi\|_s$ (and, in particular, $\|\psi\|_0\equiv\|\psi\|$ as before) hereafter;
	\item the $\dOmega$-scale $\{\focks_s\}_{s\in\mathbb{R}}$ associated with $\dOmega$, with
	\begin{equation}
		\|\Psi\|^2_{\fock_s}=\left\|\left(\dOmega+1\right)^{s/2}\Psi\right\|^2_\fock.
	\end{equation}
\end{itemize}
Note that we \textit{must} add the identity in the definition of $\|\cdot\|_{\fock_s}$ since $\dOmega$ is not strictly positive despite $\omega$ being strictly positive; indeed, $\dOmega\Omega=0$. The identity is instead unnecessary for the $\omega$-scale, since $\omega\geq m>0$; in any case, adding the identity would not affect the results hereby discussed.

The scales of Fock spaces introduced above will enable us to define GSB models with non-normalizable form factor, i.e. $f\in\hilb_{-s}\setminus\hilb$ for some $s>0$. We will show (Props.~\ref{prop:af_sing} and~\ref{prop:adagf_sing}) that, if $s\geq1$, it is possible to construct two continuous operators $\tilde{a}(f)$, $\tilde{a}^\dag(f)$, with
\begin{equation}
	\tilde{a}(f):\focks_{+s}\rightarrow\focks,\qquad 	\tilde{a}^\dag(f):\focks\rightarrow\focks_{-s},
\end{equation}
that generalize the creation and annihilation operators $\a{f}$, $\adag{f}$ in the following sense:
\begin{itemize}
	\item $\tilde{a}(f)$ acts exactly as in Eq.~\eqref{eq:af_explicit};
	\item $\adag{f}$ is its adjoint, and acts exactly as in Eq.~\eqref{eq:adagf_explicit} on $\focks_{+s}$;
	\item besides, if $f\in\hilb$, then $\tilde{a}(f)\Psi=\a{f}\Psi$ and $\tilde{a}^\dag(f)\Psi=\adag{f}\Psi$ for all $\Psi\in\focks_{+1}\subset\mathcal{D}(\a{f})$, thus their action being compatible with the ``regular'' ones introduced in Section~\ref{sec:fock}.
\end{itemize}
Let us start from the singular annihilation operator $\tilde{a}(f)$.
\begin{proposition}\label{prop:af_sing}
	Let $f\in\hilb_{-s}$ for some $s\geq1$. Then the expression ($n\geq1$)
		\begin{equation}\label{eq:af_explicit2}
		\left(\tilde{a}(f)\Psi^{(n)}\right)(k_1,\dots,k_{n-1})=\sqrt{n}\int\overline{f(k_n)}\Psi^{(n)}(k_1,\dots,k_{n-1},k_n)\,\mathrm{d}\mu(k_n),
	\end{equation}
with $\tilde{a}(f)\Omega=0$, defines a continuous map in $\mathcal{B}(\focks_{+s},\focks)$ with norm
\begin{equation}\label{eq:af_norm_s0}
	\|\tilde{a}(f)\|_{\mathcal{B}(\focks_{+s},\focks)}\leq\|f\|_{-s}.
\end{equation}
Besides, if $f\in\hilb$, for all $\Psi\in\focks_{+1}$ we have $\tilde{a}(f)\Psi=\tilde{a}(f)\Psi$.
\end{proposition}
\begin{proof} Let $\Psi\in\focks_{+s}$. Then
	\begin{eqnarray}\label{eq:estimate}
		\|\tilde{a}(f)\Psi\|^2_\fock&=&\sum_{n\geq1}\|\tilde{a}(f)\Psi^{(n)}\|^2_{\hilb^{(n-1)}}\nonumber\\
		&=&\sum_{n\geq1}n\,\int\mathrm{d}^{n-1}\mu\left|\int\mathrm{d}\mu(k_n)\,\overline{f(k_n)}\Psi^{(n)}(k_1,\dots,k_n)\right|^2\nonumber\\
		&=&\sum_{n\geq1}n\,\int\mathrm{d}^{n-1}\mu\left|\int\mathrm{d}\mu(k_n)\,\frac{\overline{f(k_n)}}{\omega(k_n)^{s/2}}\omega(k_n)^{s/2}\Psi^{(n)}(k_1,\dots,k_n)\right|^2\nonumber\\
		&\leq&\|f\|^2_{-s}\sum_{n\geq1}n\int\mathrm{d}^{n}\mu\;\,\omega(k_n)^{s}\left|\Psi^{(n)}(k_1,\dots,k_n)\right|^2\nonumber\\
		&=&\|f\|^2_{-s}\sum_{n\geq1}\int\mathrm{d}^{n}\mu\;\left(\sum_{j=1}^n\omega(k_j)^s\right)\left|\Psi^{(n)}(k_1,\dots,k_n)\right|^2\nonumber\\
		&\leq&\|f\|^2_{-s}\sum_{n\geq1}\int\mathrm{d}^{n}\mu\;\left(\sum_{j=1}^n\omega(k_j)\right)^{s}\left|\Psi^{(n)}(k_1,\dots,k_n)\right|^2\nonumber\\
		&=&\|f\|^2_{-s}\|\dOmega^{s/2}\Psi\|^2_\fock\nonumber\\
		&\leq&\|f\|^2_{-s}\|\Psi\|^2_{\fock_{+s}},
	\end{eqnarray}
where we have used the Cauchy-Schwartz inequality, the symmetry of $\Psi^{(n)}(k_1,\dots,k_n)$ under permutations of the integration variables, and the following inequality:
\begin{equation}\label{eq:algebra}
	\left(\sum_{j=1}^n c_j^s\right)\leq \left(\sum_{j=1}^n c_j\right)^s.
\end{equation}
which holds for any collection of nonnegative numbers $c_1,\dots,c_n\geq0$ and $s\geq1$. 

The last claim is immediate. Let $f\in\hilb$; then a fortiori $f\in\hilb_{-1}$. $\a{f}$ and $\tilde{a}(f)$ have the very same expression (Eqs.~\eqref{eq:af_explicit} and~\eqref{eq:af_explicit2}), and $\focks_{+1}$ is a subset of $\mathcal{D}(\a{f})$ because of Props.~\ref{prop:numb} and~\ref{prop:domaf}.
\end{proof}
It is worth pointing out that, while proven hereafter explicitly for the sake of completeness (similar computations will be made in later sections), the estimate~\eqref{eq:estimate} is in fact classic, dating back to E. Nelson, cf.~\cite[Eq. (5)]{nelson1964interaction} (see also~\cite{ginibre2006partially}).

We can now define the singular creation operator.
\begin{proposition}\label{prop:adagf_sing}
	Given $s\geq1$ and $f\in\hilb_{-s}$, there exists a unique operator $\tilde{a}^\dag(f)\in\mathcal{B}(\focks,\focks_{-s})$ such that, for all $\Phi\in\focks_{+s}$ and $\Psi\in\focks$,
	\begin{equation}\label{eq:adjoint}
\Braket{\Psi,\tilde{a}(f)\Phi}_{\fock}=\Bigl(\tilde{a}^\dag(f)\Psi,\Phi\Bigr)_{\focks_{-s},\focks_{+s}},
	\end{equation}
and its operator norm satisfies
\begin{equation}\label{eq:adagf_norm_s0}
	\|\tilde{a}^\dag(f)\|_{\mathcal{B}(\focks,\focks_{-s})}\leq\|f\|_{-s}.
\end{equation}
In particular, for all $\Psi\in\fock_{+s}$, $\tilde{a}(f)^\dag$ acts as follows ($n\geq0$):
\begin{eqnarray}\label{eq:adagf_explicit2}
	\left(\tilde{a}^\dag(f)\Psi^{(n)}\right)(k_1,\dots,k_{n},k_{n+1})&=&\frac{1}{\sqrt{n+1}}\bigg(\sum_{j=1}^n\Psi^{(n)}(k_1,\dots,\overbrace{k_{n+1}}^{j\text{th}},\dots,k_n)f(k_j)\nonumber\\&&+\Psi^{(n)}(k_1,\dots,k_n)f(k_{n+1})\bigg).
\end{eqnarray}
Finally, if $f\in\hilb$, for all $\Psi\in\focks_{+1}$ we have $\tilde{a}^\dag(f)\Psi=\adag{f}\Psi$.
\end{proposition}
\begin{proof}
Given $s\in\mathbb{R}$, the spaces $\focks_{\pm s}$, as discussed in Subsection~\ref{subsec:scale}, are mutually dual with respect to the pairing
\begin{equation}
	\Phi\in\focks_{+s},\;\Psi\in\focks_{-s}\mapsto\left(\Phi,\Psi\right)_{\fock_{+s},\fock_{-s}}:=\Braket{\left(\dOmega+1\right)^{s/2}\Phi,\left(\dOmega+1\right)^{-s/2}\Psi}_\fock.
\end{equation}
Therefore, the continuous map $\a{f}:\focks_{+s}\rightarrow\focks$ admits a unique adjoint operator with respect to this pairing, i.e. an unique continuous map from $\focks$ to $\focks_{-s}$, which we call $\tilde{a}^\dag(f)$, satisfying Eq.~\eqref{eq:adjoint}. By definition, its norm satisfies
\begin{equation}
	\|\tilde{a}^\dag(f)\|_{\mathcal{B}(\focks,\focks_{-s})}=\|\tilde{a}(f)\|_{\mathcal{B}(\focks_{+s},\focks)}\leq\|f\|_{-s}.
\end{equation}
Let us show that, given $\Psi\in\focks_{+s}$, Eq.~\eqref{eq:adagf_explicit2} holds. First of all, let us show that the right-hand side of Eq.~\eqref{eq:adagf_explicit2}, which we call $\tilde{\Psi}$, does indeed define an element of $\focks_{-s}$. We have
\begin{eqnarray}
\bigl\|\tilde{\Psi}\bigr\|^2_{\focks_{-s}}&=&\left\|\left(\dOmega+1\right)^{-s/2}\tilde{\Psi}\right\|^2_\fock\nonumber\\
	&\leq&\sum_{n\in\mathbb{N}}(n+1)\int\frac{|\Psi^{(n)}(k_1,\dots,k_n)|^2|f(k_{n+1})|^2}{\left(\sum_{\ell=1}^{n+1}\omega(k_\ell)+1\right)^s}\,\mathrm{d}^{n+1}\mu\nonumber\\
	&\leq&\sum_{n\in\mathbb{N}}(n+1)\int\frac{|\Psi^{(n)}(k_1,\dots,k_n)|^2|f(k_{n+1})|^2}{\left(\omega(k_{n+1})+1\right)^s}\,\mathrm{d}^{n+1}\mu\nonumber\\
	&\leq&\|f\|^2_{-s}\sum_{n\in\mathbb{N}}(n+1)\int|\Psi^{(n)}(k_1,\dots,k_n)|^2\,\mathrm{d}^{n+1}\mu\nonumber\\
	&=&\|f\|^2_{-s}\left(\|N^{1/2}\Psi\|_\fock^2+\|\Psi\|_\fock^2\right)<\infty,
\end{eqnarray}
where we have used the fact that, by Prop.~\ref{prop:numb}, $\Psi\in\focks_{+s}\subset\focks_{+1}=\mathcal{D}(\dOmega^{1/2})\subset\mathcal{D}(N^{1/2})$ for $s\geq1$. Therefore, the right-hand side of Eq.~\eqref{eq:adagf_explicit2} is well-defined, and a direct check shows that
	\begin{equation}\label{eq:adjointbis}
	\Braket{\Psi,\tilde{a}(f)\Phi}_{\fock}=\Bigl(\tilde{\Psi},\Phi\Bigr)_{\focks_{-s},\focks_{+s}},
\end{equation}
so that indeed $\tilde{a}^\dag(f)\Psi=\tilde{\Psi}$; Eq.~\eqref{eq:adagf_explicit2} is proven.

The final claim is thus immediate: if $f\in\hilb$, and then $f\in\hilb_{-1}$, again $\adag{f}$ is well-defined with domain $\mathcal{D}(\adag{f})\supset\focks_{+1}$ and, given $\Psi\in\focks_{+1}$, the quantities $\adag{f}\Psi$ and $\tilde{a}^\dag(f)\Psi$ coincide by a direct comparison of Eqs.~\eqref{eq:adagf_explicit} and~\eqref{eq:adagf_explicit2}.
\end{proof}

In summary, Props.~\ref{prop:af_sing}--\ref{prop:adagf_sing} enable us, whenever $f\in\hilb_{-s}$ for $s\geq1$, to define two continuous maps on the $\dOmega$-scale that can be identified as ``singular'' creation and annihilation operators.
Recalling that $\hilb\subset\hilb_{-s}\subset\hilb_{-s'}$ for all $s'>s>0$, the best possible estimate is the following one, depending on $s$:
\begin{itemize}
	\item if $f\in\hilb_{-s}$ for $s\in[0,1]$, then a fortiori $f\in\hilb_{-1}$ and thus
	\begin{equation}
		\tilde{a}(f):\focks_{+1}\rightarrow\focks,\qquad\tilde{a}^\dag(f):\fock\rightarrow\focks_{-1};
	\end{equation}
	\item if $f\in\hilb_{-s}$ for $s\geq1$, then
	\begin{equation}
		\tilde{a}(f):\focks_{+s}\rightarrow\focks,\qquad\tilde{a}^\dag(f):\fock\rightarrow\focks_{-s}.
	\end{equation}
\end{itemize}
Besides, for $f\in\hilb$, these ``singular'' operators agree with the ``regular'' ones on $\focks_{+1}$; because of that, with an abuse of notation, we will hereafter drop the tilde from them. No ambiguities will arise from this choice.

\begin{remark}\label{remark}
	Since $\focks\subset\focks_{-s}$, the operator $a(f)$ can also be interpreted as a continuous map between $\focks_{+s}$ and $\focks_{-s}$; analogously, since $\focks_{+s}\subset\focks$, the operator $a^\dag(f)$ also acts as a continuous map between $\focks_{+s}$ and $\focks_{-s}$. The operators $a(f)$ and $a^\dag(f)$ are mutually adjoint even as maps in $\mathcal{B}(\focks_{+s},\focks_{-s})$ as well, i.e.
	\begin{equation}\label{eq:adjoint2}
		\Bigl(\Psi,a(f)\Phi\Bigr)_{\focks_{+s},\focks_{-s}}=\left(a^\dag(f)\Psi,\Phi\right)_{\focks_{-s},\focks_{+s}},\qquad\forall\Phi,\Psi\in\focks_{+s},
	\end{equation}
as a direct consequence of Eq.~\eqref{eq:adjoint} and the fact that, since both $\Psi$ and $a(f)\Phi$ are in $\focks$, the left-hand side of Eq.~\eqref{eq:adjoint2} coincides with $\Braket{\Psi,a(f)\Phi}_\fock$. Besides, we also have
\begin{equation}\label{eq:norm}
		\|a(f)\|_{\mathcal{B}(\focks_{+s},\focks_{-s})}\leq\|f\|_{-s},\qquad\|a^\dag(f)\|_{\mathcal{B}(\focks_{+s},\focks_{-s})}\leq\|f\|_{-s}
\end{equation}
as a straightforward consequence of Eqs.~\eqref{eq:af_norm_s0} and~\eqref{eq:adagf_norm_s0}.

In fact, for the purposes of Section~\ref{sec:singgsb}, it would have been enough to define both operators as maps between $\focks_{+s}$ and $\focks_{-s}$ (with Eq.~\eqref{eq:adagf_explicit2} being taken as the definition of the singular creation operator); however, this would not been enough for the goals of Sections~\ref{sec:singrwa}--\ref{sec:singrwa2}.
\end{remark}

\subsection{Approximating singular creation and annihilation operators}\label{subsec:approx}
To conclude this section, we will show that every singular ($f\in\hilb_{-s}\setminus\hilb$) creation or annihilation operator $a(f),a^\dag(f)$ on the Fock scale can be approximated by a proper sequence of regular creation or annihilation operators; this property will be crucial to understand, in the next sections, the link between regular and singular GSB models.

\begin{proposition}\label{prop:singlimit}
Let $f\in\hilb_{-s}$, $s\geq1$. Then there exists a family $\{f^i\}_{i\in\mathbb{N}}\subset\hilb$ such that
\begin{equation}
	\lim_{i\to\infty}\left\|a(f)-a(f^i)\right\|_{\mathcal{B}(\focks_s,\focks)}=0,\qquad 	\lim_{i\to\infty}\left\|a^\dag(f)-a^\dag(f^i)\right\|_{\mathcal{B}(\focks,\focks_{-s})}=0,
\end{equation}
i.e. $a(f^i)\to a(f)$ and $a^\dag(f^i)\to a^\dag(f)$ in the norm sense.
\end{proposition}
\begin{proof}
	By the properties of Hilbert scales, $\hilb$ is densely embedded into $\hilb_{-s}$, implying that there exists a sequence $\{f^i\}_{n\in\mathbb{N}}\subset\hilb$ such that
	\begin{equation}
		\lim_{i\to\infty}\|f^i-f\|_{-s}=0.
	\end{equation}
	By construction, we have $a(f)-a(f^i)=a(f-f^i)$, $\adag{f}-\adag{f^i}=\adag{f-f^i}$ and thus, by Eqs.~\eqref{eq:af_norm_s0} and~\eqref{eq:adagf_norm_s0} ,
	\begin{equation}
	\|a(f)-a(f^i)\|_{\mathcal{B}(\focks_{+s},\focks)}\leq\|f-f^i\|_{-s},\qquad 	\left\|\adag{f}-a^\dag(f^i)\right\|_{\mathcal{B}(\focks,\focks_{-s})}\leq\|f-f^i\|_{-s},
	\end{equation}
henceforth the claim.
\end{proof}
\begin{remark}\label{remark2}
Recalling (see Remark~\ref{remark}) that $\a{f}$ and $\adag{f}$ can also be interpreted as continuous maps between $\focks_{+s}$ and $\focks_{-s}$, with the same operator norm (cf. Eq.~\eqref{eq:norm}), Prop.~\ref{prop:singlimit} immediately implies that, for every $f\in\hilb_{-s}$, $s\geq1$, there exists a family $\{f^i\}_{i\in\mathbb{N}}\subset\hilb$ such that
\begin{equation}
	\lim_{i\to\infty}\left\|a(f)-a(f^i)\right\|_{\mathcal{B}(\focks_{+s},\focks_{-s})}=0,\qquad 	\lim_{i\to\infty}\left\|a^\dag(f)-a^\dag(f^i)\right\|_{\mathcal{B}(\focks_{+s},\focks_{-s})}=0.
\end{equation}
\end{remark}

\section{Generalized spin-boson models with $f_1,\dots,f_r\in\hilb_{-1}$}\label{sec:singgsb}
The machinery of singular creation and annihilation operators developed in Section~\ref{sec:singcreation} will now be applied to define singular GSB models. As a first, ``zeroth-order'' application of those results, which shall then be improved in Sections~\ref{sec:singrwa}--\ref{sec:singrwa2} for a particular subclass of such models, we shall prove that GSB models with form factors $f_1,\dots,f_r\in\hilb_{-1}$ are indeed well-defined, self-adjoint operators for sufficiently values of the coupling constant $\lambda$, regardless the choice of the other parameters; furthermore, they can be approximated in the norm resolvent sense by sequences of ``regular'' GSB models with normalizable form factors. These results are collected in Prop.~\ref{prop:singgsb}.

Given the Hilbert space $\hfrak=\mathfrak{h}\otimes\focks$, with $\mathfrak{h}$ being the Hilbert space of a quantum system interacting with the boson field, let us again consider the operator $H_0$ as in Eq.~\eqref{eq:def_h0}:
\begin{equation}\label{eq:h02}
	H_0=A\otimes I+I\otimes\dOmega,
\end{equation}
with $A\in\mathcal{B}(\mathfrak{h})$ being the Hamiltonian associated to the free energy of the quantum system on $\mathfrak{h}$; as usual, we set $A\geq0$. 

Since $H_0$ is a nonnegative self-adjoint operator on $\hfrak$, following the discussion in Subsection~\ref{subsec:scale} we can construct the $H_0$-scale of Hilbert spaces $\{\hfrak_s\}_{s\in\mathbb{R}}$, the norm on $\hfrak_s$ being given by
\begin{equation}
	\|u\otimes\Psi\|_{\hfrak_s}=\left\|\left(H_0+1\right)^{s/2}(u\otimes\Psi)\right\|_\hfrak.
\end{equation}
On the other hand, we may consider as well the family of Hilbert spaces $\{\mathfrak{h}\otimes\focks_{s}\}_{s\in\mathbb{R}}$, with $\focks_s$ being the $\dOmega$-scale defined in the previous section; the corresponding norm reads
\begin{equation}
	\|u\otimes\Psi\|_{\mathfrak{h}\otimes\fock_{s}}=\|u\|_{\mathfrak{h}}\|\Psi\|_{\fock_s}=\|u\|_{\mathfrak{h}}\left\|\left(\dOmega+1\right)^{s/2}\Psi\right\|_\fock.
\end{equation}
Let us start with a simple preliminary lemma.
\begin{lemma}\label{lemma}
The spaces $\hfrak_{\pm1}$ and $\mathfrak{h}\otimes\focks_{\pm1}$ coincide.
\end{lemma}
\begin{proof}
Algebraically, both spaces $\hfrak_{+1}$ and $\mathfrak{h}\otimes\focks_{+1}$ coincide with $\mathfrak{h}\otimes\mathcal{Q}(\dOmega)$, with $\mathcal{Q}(\dOmega)$ being the form domain of $\dOmega$. To prove their equality as Hilbert spaces, the two norms $\|\cdot\|_{\hfrak_{+1}}$ and $\|\cdot\|_{\mathfrak{h}\otimes\fock_{+1}}$ must be equivalent, i.e. there must exist two constants $c_1,c_2>0$ such that
\begin{equation}\label{eq:equivnorms}
	c_1\,\|u\otimes\Psi\|_{\mathfrak{h}\otimes\fock_{+1}}\leq \|u\otimes\Psi\|_{\hfrak_{+1}}\leq c_2\,\|u\otimes\Psi\|_{\mathfrak{h}\otimes\fock_{+1}}.
\end{equation}
Now, explicitly
\begin{eqnarray}
\|u\otimes\Psi\|^2_{\hfrak_{+1}}&=&\braket{u,Au}_{\mathfrak{h}}\,\|\Psi\|^2_{\fock}+\|u\|^2_{\mathfrak{h}}\,\|\dOmega^{1/2}\Psi\|^2_{\fock}+\|u\|^2_{\mathfrak{h}}\,\|\Psi\|^2_{\fock}\nonumber\\
&=&\braket{u,Au}_{\mathfrak{h}}\,\|\Psi\|^2_{\fock}+\|u\|^2_{\mathfrak{h}}\,\|\Psi\|^2_{\fock_{+1}}\nonumber\\
&=&\braket{u,Au}_{\mathfrak{h}}\,\|\Psi\|^2_{\fock}+\|u\otimes\Psi\|^2_{\mathfrak{h}\otimes\fock_{+1}}.
\end{eqnarray}
Since we are assuming $A\geq0$, clearly $\|u\otimes\Psi\|_{\hfrak_{+1}}\geq\|u\otimes\Psi\|_{\mathfrak{h}\otimes\fock_{+1}}$; besides, since $A$ is a bounded operator on $\mathfrak{h}$,
\begin{eqnarray}
\|u\otimes\Psi\|^2_{\hfrak_{+1}}&\leq&\|A\|_{\mathcal{B}(\mathfrak{h})}\|u\|^2_{\mathfrak{h}}\|\Psi\|^2_{\fock}+\|u\otimes\Psi\|^2_{\mathfrak{h}\otimes\fock_{+1}}\nonumber\\
	&\leq&\left(1+\|A\|_{\mathcal{B}(\mathfrak{h})}\right)\|u\otimes\Psi\|^2_{\mathfrak{h}\otimes\fock_{+1}}.
\end{eqnarray}
Therefore, Eq.~\eqref{eq:equivnorms} holds with $c_1=1$ and $c_2=1+\|A\|_{\mathcal{B}(\mathfrak{h})}$. This implies that the two Hilbert spaces are equal, and so are their duals $\hfrak_{-1}$ and $\mathfrak{h}\otimes\focks_{-1}$.
\end{proof}

\begin{proposition}[Singular GSB models]\label{prop:singgsb}
	Let $H_0$ as in Eq.~\eqref{eq:h02}. The following facts holds:
	\begin{itemize}
		\item[(i)] given $f_1,\dots,f_r\in\hilb_{-1}$, $B_1,\dots,B_r\in\mathcal{B}(\mathfrak{h})$, and a coupling constant $\lambda\in\mathbb{R}$, the expression
			\begin{equation}\label{eq:def_gsb_bis}
		H_{f_1,\dots,f_r}=H_0+\lambda\sum_{j=1}^r\left(B_j\otimes a^\dag(f_j)+B_j^*\otimes a(f_j)\right)
		\end{equation}
		defines a continuous map between the Hilbert spaces $\hfrak_{+1}$ and $\hfrak_{-1}$;
		\item[(ii)] for $\lambda$ small enough, there exists $\mathcal{D}(H_{f_1,\dots,f_r})\subset\hfrak$ such that the restriction of $H_{f_1,\dots,f_r}$ to $\mathcal{D}(H_{f_1,\dots,f_r})$ defines a self-adjoint operator on $\hfrak$ with form domain 
		\begin{equation}
\mathcal{Q}(H_{f_1,\dots,f_r})=\mathcal{Q}(H_0)=\mathfrak{h}\otimes\mathcal{Q}(\dOmega);
		\end{equation}
		if $f_1,\dots,f_r\in\hilb$, said operator coincides with a regular GSB model;
		\item [(iii)] let $f_1,\dots,f_r\in\hilb_{-1}\setminus\hilb$, and $H_{f_1,\dots,f_r}$ as defined above. Then there is a family of sequences $\{f^i_1\}_{i\in\mathbb{N}},\dots,\{f^i_r\}_{n\in\mathbb{N}}$ such that
		\begin{equation}
			H_{f^i_1,\dots,f^i_r}\overset{i\to\infty}{\to}H_{f_1,\dots,f_r}\quad\text{in the norm resolvent sense.}
		\end{equation}
	\end{itemize}	
\end{proposition}
With the usual abuse of notation, the symbol $H_{f_1,\dots,f_r}$ will be used both for the continuous map between $\hfrak_{+1}$ and $\hfrak_{-1}$ defined by Eq.~\eqref{eq:def_gsb_bis}, as well as the unbounded operator on $\hfrak$ associated with it.

\begin{proof}
$(i)$ By the properties of Hilbert scales, $H_0$ can be interpreted as a continuous operator between $\hfrak_{+1}$ and $\hfrak_{-1}$. Besides, both operators $\a{f}$ and $\adag{f}$, interpreted in the sense of Props.~\ref{prop:af_sing}--\ref{prop:adagf_sing}, map continuously $\focks_{+1}$ in $\focks_{-1}$ (see Remark~\ref{remark}); since $\mathfrak{h}\otimes\focks_{\pm1}$ and $\hfrak_{\pm1}$ are isomorphic by Lemma~\ref{lemma}, the claim follows.

$(ii)$ The existence, for sufficiently small $\lambda$, of a self-adjoint operator on $\hfrak$ satisfying the desired properties, follows directly from Prop.~\ref{prop:klmn_revisited} (also see Remark~\ref{remark:klmn}), with the roles of $\mathcal{K}$, $T_0$ and $T_1$ being played respectively by $\hfrak$, $H_0$, and the map
\begin{equation}
	V_{f_1,\dots,f_r}=\sum_{j=1}^r\left(B_j^*\otimes a(f_j)+B_j\otimes a^\dag(f_j)\right).
\end{equation}
In the case in which all form factors are normalizable ($f_1,\dots,f_r\in\hilb$), the aforementioned operator coincides with the regular GSB model obtained by interpreting Eq.~\eqref{eq:def_gsb_bis} in the sense of operators on $\focks$, since, by Props.~\ref{prop:af_sing}--\ref{prop:adagf_sing}, the regular and singular creation and annihilation operators coincide on $\focks_{+1}=\mathcal{D}(\dOmega^{1/2})\supset\mathcal{D}(\dOmega)$.

$(iii)$ Because of Prop.~\ref{prop:singlimit} (also see Remark~\ref{remark2}), there exist sequences $\{f^i_1\}_{i\in\mathbb{N}},\dots,\{f^i_r\}_{n\in\mathbb{N}}$ such that, as $i\to\infty$, $a(f^i_j)\to a(f_j),a^\dag(f^i_j)\to a^\dag(f_j)$ in the sense of continuous maps between $\focks_{+1}$ and $\focks_{-1}$. This readily implies that $H_{f^i_1,\dots,f^i_r}\to H_{f_1,\dots,f_r}$ in the sense of continuous maps between $\hfrak_{+1}$ and $\hfrak_{-1}$. By~\cite[Theorem~VIII.25]{reed1972methods}, this implies convergence in the norm resolvent sense.
\end{proof}

We have thus defined a family of self-adjoint operators on $\focks$, depending on a family of functions $f_1,\dots,f_r\in\hilb_{-1}\supset\hilb$, which does include ``regular'' GSB models ($f_1,\dots,f_r\in\hilb$) as a special case, and can be approximated by them in the norm resolvent topology; as such, they are the ``correct'' generalization of GSB models, hence justifying our nomenclature. Importantly, the formalism of singular creation and annihilation operators gives a precise mathematical meaning to the formal expressions analogous to Eq.~\eqref{eq:def_gsb_bis} often encountered in the physical literature.

Some remarks are in order. First of all, norm resolvent convergence is a powerful notion of convergence for unbounded self-adjoint operators: for one, it ensures that the unitary evolution group generated by a singular GSB model can be approximated, in the strong sense, by the evolution groups generated by a proper sequence of regular models~\cite{reed1972methods,teschl2009mathematical,de2008intermediate,derezinski2013unbounded}; furthermore, as a consequence of norm resolvent convergence, the spectral properties of singular GSB models are largely inherited by those of regular GSB models. We will leave a detailed study of these questions to future works.

We point out that, as discussed in Remark~\ref{remark:klmn} for the general case, Prop.~\ref{prop:singgsb} does not give us information about the operator domain of a singular GSB model, which, differently from the regular case, will depend nontrivially on the form factors $f_1,\dots,f_r$ as well as the coupling constant $\lambda$; nevertheless, the form domain is still the same as the one in the regular case. 

Finally, we remark that this is a perturbative result: it only holds for sufficiently small values of $\lambda$, and thus, in principle, may improved via different techniques. In this spirit, nonperturbative results for a specific class of GSB models will be analyzed in Sections~\ref{sec:singrwa}--\ref{sec:singrwa2}.

\section{The rotating-wave spin-boson model with $f\in\hilb_{-1}$}\label{sec:singrwa}

We will hereafter focus on the rotating-wave (RW) spin-boson model, cf.~Eq.~\eqref{eq:def_sbrwa}, for which the formalism of Hilbert scales turns out to be particularly useful. In order to keep the discussion simple, we will mostly deal with the case of a single two-level system.

Following a different strategy, we will show that, given $f\in\hilb_{-1}$, it is possible to define a self-adjoint operator on $\mathbb{C}^2\otimes\focks\cong\focks\oplus\focks$ which can be obtained as the norm resolvent limit of a sequence of RW spin-boson models with form factor $f^i\in\hilb$; as such, it represents the correct extension of the RW spin-boson model to a form factor $f\in\hilb_{-1}$, that is, a ``singular'' RW spin-boson model. This result improves the general one of Section~\ref{sec:singgsb} in two directions:
\begin{itemize}
	\item it is a nonperturbative result: it holds for every value of the coupling constant $\lambda$, and not just for sufficiently small values;
	\item by construction, it allows for an explicit evaluation of the resolvent (thus allowing for a direct study of the spectral properties of the model) as well as the operator domain.
\end{itemize}
A further extension of these results to ``more singular'' form factors, up to $f\in\hilb_{-2}$, will be discussed in Section~\ref{sec:singrwa2}. This nontrivial refinement will be obtained via an interesting \textit{renormalization} procedure.

We will start our analysis by investigating, in Subsection~\ref{subsec:regular}, the mathematical properties of the model with form factor $f\in\hilb$ and its decomposition on sectors with a fixed number of excitations. Its extension to a non-normalizable form factor $f\in\hilb_{-1}$ (Theorem~\ref{thm:singrwa}) will be presented and discussed in Subsection~\ref{subsec:result}, along with some remarks about its possible generalizations to the many-atom case.

\subsection{The rotating-wave spin-boson model}\label{subsec:regular}   
Let us start by analyzing in greater detail the structure of the rotating-wave spin-boson model in the regular case. As already discussed in Subsection~\ref{subsec:gsb}, given $f\in\hilb$, the model is defined on $\mathbb{C}^2\otimes\focks$ via
\begin{equation}\label{eq:def_sbrwa_2}
H_{f}=H_0+\lambda\left(\sigma_+\otimes a(f)+\sigma_-\otimes a^\dag(f)\right).
\end{equation}
Without loss of generality, we will set the ground state energy $\omega_{\mathrm{g}}$ to zero hereafter, so that
\begin{equation}
	H_0=\begin{pmatrix}
		\omega_{\mathrm{e}}&0\\
		0&0
	\end{pmatrix}\otimes I+I\otimes\dOmega.
\end{equation}
Since $\hfrak=\mathbb{C}^2\otimes\focks\simeq\focks\oplus\focks$, we can write the most general element of the total Hilbert space as a column vector:
\begin{equation}
\begin{pmatrix}
\Psie\\\Psig
\end{pmatrix},\qquad\Psie,\Psig\in\focks,
\end{equation}
with $\Psie,\Psig$ being the states of the boson field when the atom is respectively in its excited and ground state; in particular, the states $(\Omega,0)^\intercal$ and $(0,\Omega)^\intercal$ are the states in which the atom is respectively in its excited and ground state, and the boson field is in its vacuum state $\Omega$. We shall work in this representation hereafter. The model can thus be written in a formal matrix fashion:
\begin{equation}\label{eq:matrix}
H_{f,\omega_\e}=\begin{pmatrix}
\omega_{\mathrm{e}}+\dOmega&\lambda\,\a{f}\\
\lambda\,\adag{f}&\dOmega
\end{pmatrix},
\end{equation}
with domain $\mathcal{D}(H_{f,\omega_\e})\simeq\mathcal{D}(\dOmega)\oplus\mathcal{D}(\dOmega)$; for future convenience, we will hereafter indicate explicitly the dependence of the model on $\omega_\e$.

A peculiar feature of this model, which makes it particularly easy to study, is the following one: the model preserves the total number of excitations of the system, thus being decomposed into a direct sum. Let us elaborate on that.
\begin{proposition}\label{prop:numexc}
	Let $\hfrak^{(0)}=\{0\}\oplus\hilb^{(0)}$ and, for all $n\geq1$, let $\hfrak^{(n)}=\hilb^{(n-1)}\oplus\hilb^{(n)}$. Then $\hfrak^{(n)}$ is a reducing subspace for $H_{f,\omega_\e}$, and
	\begin{equation}\label{eq:decom}
	H_{f,\omega_\e}=\bigoplus_{n\in\mathbb{N}}H_{f,\omega_\e}^{(n)},
	\end{equation}
	with $H_{f,\omega_\e}^{(n)}$ being the restriction of $H_{f,\omega_\e}$ to $\hfrak^{(n)}$.
\end{proposition}
\begin{proof}
	The claim is obvious for $n=0$, so let $n\geq1$. The most generic element of $\hfrak^{(n)}$ can be written as $(\Psie^{(n-1)},\Psig^{(n)})^\intercal$, with $\Psie^{(n-1)}\in\hilb^{(n-1)}$ and $\Psig^{(n)}\in\hilbn$. By construction, if $\Psie^{(n-1)}\in\mathcal{D}(\omega^{(n-1)})$, $\Psig^{(n)}\in\mathcal{D}(\omega^{(n)})$, we have\begin{equation}
	H_{f,\omega_\e}\begin{pmatrix}
	\Psie^{(n-1)}\\\Psig^{(n)}
	\end{pmatrix}=\begin{pmatrix}
	(\omega_{\mathrm{e}}+\dOmega)\Psie^{(n-1)}+\lambda\,a(f)\Psig^{(n)}\\
	\dOmega\Psig^{(n)}+\lambda\,a^\dag(f)\Psie^{(n-1)}
	\end{pmatrix},
	\end{equation}
	and the claim follows by the known properties of $\dOmega,\a{f}$, and $\adag{f}$ (see Eq.~\eqref{eq:ladder}). Since $\hfrak=\oplus_{n\in\mathbb{N}}\hfrak^{(n)}$, $H_{f,\omega_\e}$ is thus decomposed as in Eq.~\eqref{eq:decom}.
\end{proof}
Necessarily, each operator $H_{f,\omega_\e}^{(n)}$ is self-adjoint on $\hfrak^{(n)}$. Physically, $\hfrak^{(n)}$ can be interpreted as the subspace of all states in $\hfrak$ that have $n$ excitations, i.e. either the atom is its ground state and there are $n$ bosons in the field, or the atom is in its excited state and there are $n-1$ bosons in the field. Indeed, the interaction term in Eq.~\eqref{eq:def_sbrwa_2} is made in such a way to implement one of the following transitions:
\begin{itemize}
	\item the atom switches from the excited to the ground state, and a boson with wavefunction $f$ is created in the process;
	\item the atom switches from the ground to the excited state, and a boson with wavefunction $f$ is annihilated in the process.
\end{itemize}
Mathematically, $\hfrak^{(n)}$ is the $n$th eigenspace of the operator
\begin{equation}
	N_{\mathrm{exc}}=\begin{pmatrix}
	N+1&0\\
	0&N
\end{pmatrix}.
\end{equation}
In particular, on the single-excitation sector $\hfrak^{(1)}=\mathbb{C}\oplus\hilb$, the model acts as
\begin{equation}\label{eq:fl}
H_{f,\omega_\e}^{(1)}=\begin{pmatrix}
\omega_{\mathrm{e}}&\lambda\braket{f,\cdot}\\
\lambda\,f&\omega
\end{pmatrix},
\end{equation}
and corresponds to a Friedrichs (or Friedrichs-Lee) model~\cite{gadella2011friedrichs,lakaev1989some,facchi2021spectral}; remarkably, a singular version of $H_{f,\omega_\e}^{(1)}$, accommodating a form factor up to $f\in\hilb_{-2}$, has been indeed constructed~\cite{derezinski2002renormalization,facchi2021spectral}.

\begin{remark}\label{remark:natom}
	The $r$-atom generalization of this model can be investigated similarly; as an example, let us briefly discuss the case $r=2$. Now $\hfrak=\mathbb{C}^{2}\otimes\mathbb{C}^2\otimes\focks\simeq\oplus_{j=1}^4\focks$, and the most general element of the total Hilbert space can be written as
	\begin{equation}\label{eq:4}\left(
		\begin{array}{c}
			\Psi_{\mathrm{ee}}\\\hline
			\Psi_{\mathrm{eg}}\\
			\Psi_{\mathrm{ge}}\\\hline
			\Psi_{\mathrm{gg}}
		\end{array}\right),\qquad\Psi_{\mathrm{ee}},\Psi_{\mathrm{eg}},\Psi_{\mathrm{ge}},\Psi_{\mathrm{gg}}\in\focks,
	\end{equation}
where $\Psi_{xx'}$, $x,x'\in\{\mathrm{e,g}\}$ corresponds to the states of the boson field when the first and the second atom are respectively in the $x$ and $x'$ state. Given $f_1,f_2\in\hilb$, again the model can be written in a matrix fashion similar to Eq.~\eqref{eq:matrix}, namely
\begin{equation}\label{eq:matrix2}
	H_{f_1,f_2}=\left(\begin{array}{c|cc|c}
		H_{\mathrm{ee}}&\lambda\,\a{f_2}&\lambda\,\a{f_1}&0\\\hline
		\lambda\,\adag{f_2}&H_{\mathrm{eg}}&0&\lambda\,\a{f_1}\\
		\lambda\,\adag{f_1}&0&H_{\mathrm{ge}}&\lambda\,\a{f_2}\\\hline
		0&\lambda\,\adag{f_1}&\lambda\,\adag{f_2}&H_{\mathrm{gg}}
	\end{array}\right),
\end{equation}
where, for brevity, $H_{xx'}=\omega_{x,1}+\omega_{x',2}+\dOmega$ for $x,x'\in\{\mathrm{e,g}\}$, with $\omega_{x,j}$ being the energy of the $j$th atom in its $x$ state; the domain of $H_{f_1,f_2}$ is given by $\mathcal{D}(H_{f_1,f_2})=\oplus_{j=1}^4\mathcal{D}(\dOmega)$. In both Eqs.~\eqref{eq:4}--\eqref{eq:matrix2}, we have stressed the distinction between ``sectors'' with an equal number of atoms in the excited state. Finally, a decomposition analogous to the one discussed in Prop.~\ref{prop:numexc} may be found as well.
\end{remark}

\subsection{Extension of the model to form factors $f\in\hilb_{-1}$}\label{subsec:result}
Let $f\in\hilb_{-1}$, and let us consider again the following expression,
\begin{equation}\label{eq:formal}
H_{f,\omega_\e}=\begin{pmatrix}
\omega_{\mathrm{e}}+\dOmega&\lambda\,\a{f}\\
\lambda\,\adag{f}&\dOmega
\end{pmatrix}.
\end{equation}
When $f\in\hilb_{-1}\setminus\hilb$, such an expression cannot obviously define an operator on $\focks\oplus\focks$, since $\adag{f}$ has values outside the Fock space $\focks$; still, as already seen in the general case of GSB models, it does define a continuous operator between the Hilbert spaces $\focks_{+1}\oplus\focks_{+1}$ and $\focks_{-1}\oplus\focks_{-1}$. Since, by Lemma~\ref{lemma},
\begin{equation}
	\hfrak_{\pm1}\simeq\mathbb{C}^2\otimes\focks_{\pm1}\simeq\focks_{\pm1}\oplus\focks_{\pm1},
\end{equation}
with $\{\hfrak_s\}_{s\in\mathbb{R}}$ being the scale of Fock spaces associated with $H_0$, then for every choice of $f\in\hilb_{-1}$ the expression in Eq.~\eqref{eq:formal} defines a continuous operator between $\hfrak_{+1}$ and $\hfrak_{-1}$.

We wonder whether we can interpret $H_{f,\omega_\e}$ as a self-adjoint operator on $\focks\oplus\focks$, i.e. whether we can find a self-adjointness domain in $\focks\oplus\focks$ for it. The existence of such a domain (for small coupling) is ensured by Prop.~\ref{prop:singgsb}; here, however, we will be able to find explicitly such a domain without requiring the coupling constant to be small.

Let us start from a preliminary lemma.
\begin{lemma}\label{lemma2}
Let $f\in\hilb_{-1}$, $\omega_\e,\lambda\in\mathbb{R}$, and define the operators\footnote{Strictly speaking, we should also exclude $z=0$ from the definition of $\mathcal{S}_f(z)$ in Eq.~\eqref{eq:sf}, since $\dOmega$ is not invertible because of $\dOmega\Omega=0$, with $\Omega$ being the vacuum state. However, the presence of the creation operator at the right in the definition of $\mathcal{S}_f(z)$ allows us to extend the definition of $\mathcal{S}_f(z)$ to $z=0$.}
\begin{eqnarray}\label{eq:gf}
	\mathcal{S}_f(z)&=&\a{f}\frac{1}{\dOmega-z}\adag{f},\qquad z\in\mathbb{C}\setminus[m,\infty);\\\label{eq:sf} \mathcal{G}_{f,\omega_\e}(z)&=&\omega_{\mathrm{e}}-z+\dOmega-\lambda^2\mathcal{S}_f(z),\qquad z\in\mathbb{C}\setminus\mathbb{R},
\end{eqnarray}
with $\a{f}$, $\adag{f}$ to be interpreted in the sense of Props.~\ref{prop:af_sing}--\ref{prop:adagf_sing}. Then
\begin{itemize}
	\item $\mathcal{S}_f(z)$ is a bounded operator on $\focks$, with $\mathcal{S}_f(\bar{z})=\mathcal{S}_f(z)^*$;
	\item $\mathcal{G}_{f,\omega_\e}(z)$, with domain $\mathcal{D}(\mathcal{G}_{f,\omega_\e}(z))=\mathcal{D}(\dOmega)$, is a closed operator on $\focks$ satisfying $\mathcal{G}_{f,\omega_\e}(\bar{z})=\mathcal{G}_{f,\omega_\e}(z)^*$ and admitting a bounded inverse $\mathcal{G}_{f,\omega_\e}^{-1}(z)$ with operator norm
	\begin{equation}\label{eq:boundednorm1}
		\left\|\mathcal{G}_{f,\omega_\e}^{-1}(z)\right\|_{\mathcal{B}(\focks)}\leq\frac{1}{|\Im z|}.
	\end{equation}
\end{itemize}
\end{lemma}
We will refer to $\mathcal{G}_{f,\omega_\e}^{-1}(z)$ as the \textit{propagator} of the model; its fundamental role will be clear momentarily.
\begin{proof}
	By Prop.~\ref{prop:adagf_sing}, $\adag{f}$ maps continuously $\focks$ in $\focks_{-1}$; by the standard properties of Hilbert scales and the fact that any nonreal $z$ belongs to the resolvent of $\dOmega$, $(\dOmega-z)^{-1}$ maps continuously $\focks_{-1}$ in $\focks_{+1}$; finally, by Prop.~\ref{prop:af_sing}, $\a{f}$ maps continuously $\focks_{+1}$ in $\focks$, hence $\mathcal{S}_f(z)$ is a bounded operator on $\focks$, and the property $\mathcal{S}_f(z)^*=\mathcal{S}_f(\bar{z})$ is immediate. Consequently, $\mathcal{G}_{f,\omega_\e}(z)$ is a well-defined closed operator with domain $\mathcal{D}(\dOmega)$, since it is simply obtained by summing a bounded operator to $\dOmega$, and satisfies $\mathcal{G}_{f,\omega_\e}(z)^*=\mathcal{G}_{f,\omega_\e}(\bar{z})$ as well.
	
	We must show that $\mathcal{G}_{f,\omega_\e}(z)$ admits a bounded inverse. This happens if and only if there is some $c>0$ such that, for all $\Psi\in\mathcal{D}(\dOmega)$,
		\begin{equation}\label{eq:dere}
				\|\mathcal{G}_{f,\omega_\e}(z)\Psi\|_\fock\geq c\|\Psi\|_\fock\quad\text{and}\quad\|\mathcal{G}_{f,\omega_\e}(z)^*\Psi\|_\fock\geq c\|\Psi\|_\fock,
			\end{equation}
		see e.g.~\cite[Theorem~3.3.2]{derezinski2013unbounded}. Now, for all $0\neq\Psi\in\mathcal{D}(\dOmega)$,
		\begin{eqnarray}
				\Im	\Braket{\Psi,\mathcal{S}_f(z)\Psi}_\fock&=&\Im z\Braket{\Psi,a(f)\frac{1}{\dOmega-z}\frac{1}{\dOmega-\bar{z}}a^\dag(f)\Psi}_\fock\nonumber\\
				&=&\Im z\left\|\frac{1}{\dOmega-\bar{z}}a^\dag(f)\Psi\right\|_\fock^2,
			\end{eqnarray}
		therefore
		\begin{eqnarray}
	\Im	\Braket{\Psi,\mathcal{G}_{f,\omega_\e}(z)\Psi}_\focks=-\Im z\left(\|\Psi\|_{\fock}^2+\lambda^2\left\|\frac{1}{\dOmega-\bar{z}}a^\dag(f)\Psi\right\|_\fock^2\right)
		\end{eqnarray}
	implying
		\begin{eqnarray}
				|\Braket{\Psi,\mathcal{G}_{f,\omega_\e}(z)\Psi}_\fock|\geq|\Im\Braket{\Psi,\mathcal{G}_{f,\omega_\e}(z)\Psi}_\fock|\geq|\Im z|\|\Psi\|^2_\fock>0,
			\end{eqnarray}
		also implying, by the Cauchy-Schwartz inequality,
		\begin{equation}
				\|\mathcal{G}_{f,\omega_\e}(z)\Psi\|_\fock\geq|\Im z|\|\Psi\|_\fock
			\end{equation}
		and thus, since $\mathcal{G}_{f,\omega_\e}(z)^*=\mathcal{G}_{f,\omega_\e}(\bar{z})$,
		\begin{equation}
				\|\mathcal{G}_{f,\omega_\e}(z)^*\Psi\|_\fock\geq|\Im z|\|\Psi\|_\fock,
			\end{equation}
		i.e. Eq.~\eqref{eq:dere} holds with $c=|\Im z|$, finally implying that $\mathcal{G}_{f,\omega_\e}(z)$ admits a bounded inverse in $\focks$ with operator norm satisfying Eq.~\eqref{eq:boundednorm1}.
	\end{proof}
\begin{remark}\label{remark:self}
We can also compute explicitly the action of $\mathcal{S}_f(z)$ on $\focks_{+1}$ by employing Eqs.~\eqref{eq:af_explicit2} and~\eqref{eq:adagf_explicit2}: given $\Psi\in\focks_{+1}$, we have
\begin{equation}
	\mathcal{S}_f(z)\Psi=\left(\bigoplus_{n\in\mathbb{N}}\mathcal{S}_f^{(n)}(z)\right)\Psi,
\end{equation}
where
	\begin{eqnarray}\label{eq:sfn}
			\left(\mathcal{S}_f^{(n)}(z)\Psi^{(n)}\right)(k_1,\dots,k_n)\!&=&\!\sum_{j=1}^n\left(\int\mathrm{d}\mu(\kappa)\,\frac{\overline{f(\kappa)}\Psi^{(n)}(k_1,\dots,\overbrace{\kappa}^{j\text{th}},\dots,k_n)}{\omega(\kappa)+\sum_{\ell=1}^n\omega(k_\ell)-z}\right)\!f(k_j)\nonumber\\&&+\!\left(\int\mathrm{d}\mu(\kappa)\frac{|f(\kappa)|^2}{\omega(\kappa)+\sum_{j=1}^n\omega(k_j)-z}\right)\!\Psi^{(n)}(k_1,\dots,k_n).
	\end{eqnarray}
In particular, for $n=0$, we simply have
\begin{equation}\label{eq:self}
	\mathcal{S}_f^{(0)}(z)\equiv\Sigma_f(z)=\int\frac{|f(k)|^2}{\omega(k)-z}\,\mathrm{d}\mu(k),
\end{equation}
which coincides with the (non-renormalized) self-energy of the single-excitation sector $H_{f,\omega_\e}^{(1)}$ of the model~\cite{facchi2021spectral,lonigro2021selfenergy}. Eq.~\eqref{eq:sfn} will play a crucial role in Section~\ref{sec:singrwa2}.
\end{remark}
We are now ready to state our main result of this section.
\begin{theorem}\label{thm:singrwa}
Let $f\in\hilb_{-1}$, and let $H_{f,\omega_\e}$ be the operator on $\focks$ with domain
\begin{equation}\label{eq:singdom}
	\mathcal{D}(H_{f,\omega_\e})=\left\{
	\begin{pmatrix}
	\Phie\\\Phig-\lambda\frac{1}{\dOmega+1}\adag{f}\Phie
	\end{pmatrix}:\;\Phie,\Phig\in\mathcal{D}\!\left(\dOmega\right)
	\right\},
\end{equation}
acting as\renewcommand\arraystretch{1.5}
\begin{equation}\label{eq:action}
	H_{f,\omega_\e}\begin{pmatrix}
	\Phie\\\Phig-\lambda\frac{1}{\dOmega+1}\adag{f}\Phie
	\end{pmatrix}=\begin{pmatrix}
	\left(\omega_{\mathrm{e}}+\dOmega-\lambda^2\a{f}\!\frac{1}{\dOmega+1}\adag{f}\right)\Phie+\lambda\,\a{f}\Phig\\
	\dOmega\Phig+\lambda\frac{1}{\dOmega+1}\adag{f}\Phie
	\end{pmatrix}.
\end{equation}
Then the following facts hold for every value of $\omega_\e,\lambda\in\mathbb{R}$:
\begin{enumerate}
	\item [(i)] for $f\in\hilb$, $H_{f,\omega_\e}$ coincides with the (regular) rotating-wave spin-boson model;
	\item [(ii)] for $f\in\hilb_{-1}$, $H_{f,\omega_\e}$ is a self-adjoint operator on $\focks$ whose resolvent reads, for all $z\in\mathbb{C}\setminus\mathbb{R}$,
	\begin{equation}\label{eq:ressing}
	\frac{1}{H_{f,\omega_\e}-z}\begin{pmatrix}
	\Psie\\\Psig
	\end{pmatrix}=\begin{pmatrix}
	\mathcal{G}_{f,\omega_\e}^{-1}(z)\Bigl(\Psie-\lambda\,\a{f}\frac{1}{\dOmega-z}\Psig\Bigr)\\
	\frac{1}{\dOmega-z}\Psig-\lambda\frac{1}{\dOmega-z}\adag{f}\mathcal{G}^{-1}_{f,\omega_\e}(z)\Bigl(\Psie-\lambda\,\a{f}\frac{1}{\dOmega-z}\Psig\Bigr)
	\end{pmatrix},
	\end{equation}
	with $\mathcal{G}_{f,\omega_\e}(z)$ as defined in Eq.~\eqref{eq:gf};
	\item [(iii)] given $f\in\hilb_{-1}\setminus\hilb$, there is a sequence $\{f^i\}_{i\in\mathbb{N}}\subset\hilb$ of normalizable form factors such that $H_{f^i,\omega_\e}\to H_{f,\omega_\e}$ in the norm resolvent sense;
	\item [(iv)] conversely, given any $\{f^i\}_{i\in\mathbb{N}}\subset\hilb$, $f\in\hilb_{-1}\setminus\hilb$ such that $\|f^i-f\|_{-1}\to0$, then the sequence of regular rotating-wave spin-boson models $\{H_{f^i,\omega_\e}\}_{i\in\mathbb{N}}$ converges to the operator $H_{f,\omega_\e}$ in the norm resolvent sense.
\end{enumerate}
\end{theorem}
\begin{proof}
$(i)$ If $f\in\hilb$, then, by Prop.~\ref{prop:adagf_sing}, we know that $\adag{f}\Phie\in\focks$, and then $(\dOmega+1)^{-1}$ maps it back into $\mathcal{D}(\dOmega)$. This means that, as long as $f\in\hilb$, Eq.~\eqref{eq:singdom} is nothing but an alternative, and equivalent, representation of the domain $\mathcal{D}(H_{f,\omega_\e})=\mathcal{D}(\dOmega)\oplus\mathcal{D}(\dOmega)$ of the regular model, and a direct computation shows that, indeed, the quantity
\begin{equation}
	\begin{pmatrix}
	\omega_{\mathrm{e}}+\dOmega&\lambda\,\a{f}\\
	\lambda\,\adag{f}&\dOmega
	\end{pmatrix}\begin{pmatrix}
	\Phie\\\Phig-\lambda\frac{1}{\dOmega+1}\adag{f}\Phie
	\end{pmatrix}
\end{equation}
equals the right-hand side of Eq.~\eqref{eq:action}.

$(ii)$ Let $f\in\hilb_{-1}$. By construction, $\mathcal{D}(H_{f,\omega_\e})$ is dense in $\focks\oplus\focks$. To compute its resolvent, we must solve the equation\renewcommand\arraystretch{1}
\begin{equation}
	(H_{f,\omega_\e}-z)\begin{pmatrix}
	\Phie\\\Phig-\lambda\frac{1}{\dOmega+1}\adag{f}\Phie
	\end{pmatrix}=\begin{pmatrix}
	\Psie\\\Psig
	\end{pmatrix}
\end{equation}
for $z\in\mathbb{C}\setminus\mathbb{R}$, that is,\renewcommand\arraystretch{1.5}
\begin{equation}
\begin{pmatrix}	\left(\omega_{\mathrm{e}}+\dOmega-z-\lambda^2\a{f}\!\frac{1}{\dOmega+1}\adag{f}\right)\Phie+\lambda\,\a{f}\Phig\\
	\left(\dOmega-z\right)\Phig+(z+1)\lambda\frac{1}{\dOmega+1}\adag{f}\Phie
\end{pmatrix}=\begin{pmatrix}
\Psie\\\Psig
\end{pmatrix}.
\end{equation}\renewcommand\arraystretch{1}
The second equation yields
\begin{eqnarray}\label{eq:phig}
	\Phig&=&\frac{1}{\dOmega-z}\Psig-(z+1)\lambda\,\frac{1}{\dOmega-z}\frac{1}{\dOmega+1}\adag{f}\Phie\nonumber\\
	&=&\frac{1}{\dOmega-z}\Psig-\lambda\left(\frac{1}{\dOmega-z}-\frac{1}{\dOmega+1}\right)\adag{f}\Phie
\end{eqnarray}
Substituting into the first one, we get
\begin{equation}
\Bigl(\omega_{\mathrm{e}}+\dOmega-z-\lambda^2\mathcal{S}_f(z)\Bigr)\Phie+\lambda\,\a{f}\frac{1}{\dOmega-z}\Psig=\Psie
\end{equation}
with $\mathcal{S}_f(z)$ as in Eq.~\eqref{eq:sf}, that is,
\begin{equation}
	\mathcal{G}_{f,\omega_\e}(z)\Phie=\Psie-\lambda\,\a{f}\frac{1}{\dOmega-z}\Psig.
\end{equation}
By Lemma~\ref{lemma2}, $\mathcal{G}_{f,\omega_\e}(z)$ admits a bounded inverse and therefore
\begin{equation}\label{eq:phie}
	\Phie=\mathcal{G}^{-1}_{f,\omega_\e}(z)\left(\Psie-\lambda\,\a{f}\frac{1}{\dOmega-z}\Psig\right),
\end{equation}
which is the first component of the right-hand side in Eq.~\eqref{eq:ressing}. Substituting Eq.~\eqref{eq:phie} into~\eqref{eq:phig} finally yields
\begin{eqnarray}
\Phig-\lambda\frac{1}{\dOmega+1}\adag{f}\Phie&=&\frac{1}{\dOmega-z}\Psig\\&&-\lambda\frac{1}{\dOmega-z}\adag{f}\mathcal{G}^{-1}_{f,\omega_\e}(z)\left(\Psie-\lambda\,\a{f}\frac{1}{\dOmega-z}\Psig\right),\nonumber
\end{eqnarray}
which is the second component of the right-hand side in Eq.~\eqref{eq:ressing}. Eq.~\eqref{eq:ressing} is proven for all nonreal $z$. By construction, $H_{f,\omega_\e}$ is therefore a self-adjoint operator on $\focks\oplus\focks$.

$(iii)$ Let $f\in\hilb_{-1}$. Since $\hilb$ is densely embedded into $\hilb_{-1}$, there exists a sequence $\{f^i\}_{i\in\mathbb{N}}$ such that $\|f^i-f\|_{-1}\to0$ as $i\to\infty$, and thus, by Prop.~\ref{prop:singlimit}, $a(f^i)\to\a{f}$ and $a(f^i)\to\adag{f}$ in norm; this readily implies that $\mathcal{S}_{f^i}(z)\to\mathcal{S}_{f}(z)$ in the norm sense. But then we have
	\begin{eqnarray}
		\mathcal{G}^{-1}_{f,\omega_\e}(z)-\mathcal{G}^{-1}_{f^i,\omega_\e}(z)&=&	\mathcal{G}^{-1}_{f^i,\omega_\e}(z)\left[\mathcal{G}_{f^i,\omega_\e}(z)-\mathcal{G}_{f,\omega_\e}(z)\right]\mathcal{G}^{-1}_{f,\omega_\e}(z)\nonumber\\
		&=&\lambda^2\mathcal{G}^{-1}_{f^i,\omega_\e}(z)\left[\mathcal{S}_{f^i}(z)-\mathcal{S}_{f}(z)\right]\mathcal{G}^{-1}_{f,\omega_\e}(z),
	\end{eqnarray}
	whence, again using Lemma~\ref{lemma2}
	\begin{eqnarray}
		\left\|	\mathcal{G}^{-1}_{f,\omega_\e}(z)-\mathcal{G}^{-1}_{f^i,\omega_\e}(z)\right\|_{\mathcal{B}(\focks)}
		&\leq&\lambda^2	\left\|\mathcal{G}^{-1}_{f^i,\omega_\e}(z)\right\|_{\mathcal{B}(\focks)}\left\|\left[\mathcal{S}_{f^i}(z)-\mathcal{S}_{f}(z)\right]\mathcal{G}^{-1}_{f,\omega_\e}(z)\right\|_{\mathcal{B}(\focks)}\nonumber\\
		&\leq&\frac{\lambda^2}{|\Im z|}\left\|\left[\mathcal{S}_{f^i}(z)-\mathcal{S}_{f}(z)\right]\mathcal{G}^{-1}_{f,\omega_\e}(z)\right\|_{\mathcal{B}(\focks)}\to0,
	\end{eqnarray}
which implies norm resolvent convergence. $(iv)$ is proven analogously.
\end{proof}
Notice that Eq.~\eqref{eq:ressing}, practically speaking, means that the properties of the resolvent are entirely encoded in those of the propagator $\mathcal{G}_{f,\omega_\e}^{-1}(z)$.
\begin{remark}\label{rem:vacuum}
We remark that Eq.~\eqref{eq:singdom} provides just one of the possible representations of the domain $\mathcal{D}(H_{f,\omega_\e})$ of the singular model: we may equivalently write
\begin{equation}\label{eq:singdom2}
	\mathcal{D}(H_{f,\omega_\e})=\left\{
	\begin{pmatrix}
		\Phie\\\Phig-\lambda\frac{1}{\dOmega-z_0}\adag{f}\Phie
	\end{pmatrix}:\;\Phie,\Phig\in\mathcal{D}\!\left(\dOmega\right).
	\right\}
\end{equation}
for any fixed $z_0\in\mathbb{C}$ which belongs to the resolvent of $\dOmega$: Eq.~\eqref{eq:singdom} simply corresponds to the choice $z_0=-1$. Choosing a different $z_0$, Eq.~\eqref{eq:action} must be changed as well accordingly. 

What is really important, and independent of the particular choice of $z_0$, is the following observation: while for $f\in\hilb$ the two component of a state $\Psi\in\focks\oplus\focks$ can be chosen independently, in the singular case they must be ``coupled'': the ground component must have a coupling-dependent singular part which depends on the excited component. The role of such an additional term is to cancel out the ``divergent'' term $\adag{f}\Phie$, which does not belong to $\focks$ whenever $f\notin\hilb$. In particular, the vector\renewcommand\arraystretch{1}
\begin{equation}\label{eq:vacuum}
	\Psi_0=\begin{pmatrix}
		\Omega\\0
	\end{pmatrix},
\end{equation}\renewcommand\arraystretch{1.5}corresponding to the state in which the atom is excited and the boson field is in the vacuum, is not in $\mathcal{D}(H_{f,\omega_\e})$ whenever $f\in\hilb_{-1}\setminus\hilb$. Physically, this means that the total energy distribution of such a state has an infinite variance.
\end{remark}
Before discussing the refinement of these results to the ``more singular'' case $f\in\hilb_{-2}$, let us present a simple corollary of Theorem~\ref{thm:singrwa} concerning the restriction of the singular model to the $n$-excitation subspaces $\hfrak^{(n)}$, its proof being immediate.
\begin{corollary}\label{coroll1}
	Let $f\in\hilb_{-1}$. For all $n\in\mathbb{N}$, the restriction $H_{f,\omega_\e}^{(n)}$ of $H_{f,\omega_\e}$ to the $n$-excitation subspace $\hfrak^{(n)}$ is a self-adjoint operator on $\hfrak^{(n)}$, with domain
	\begin{equation}\label{eq:singdom_n}
	\mathcal{D}\left(H_{f,\omega_\e}^{(n)}\right)=\left\{
	\begin{pmatrix}
	\Phie^{(n-1)}\\\Phig^{(n)}-\lambda\frac{1}{\omega^{(n)}+1}\adag{f}\Phie^{(n-1)}
	\end{pmatrix}:\;\Phie^{(n-1)}\in\mathcal{D}(\omega^{(n-1)}),\;\Phig^{(n)}\in\mathcal{D}(\omega^{(n)}),	
	\right\},
	\end{equation}
	acting as
	\begin{equation}
	H_{f,\omega_\e}^{(n)}\begin{pmatrix}
	\Phie^{(n-1)}\\\Phig^{(n)}-\lambda\frac{1}{\omega^{(n)}+1}\adag{f}\Phie^{(n-1)}
	\end{pmatrix}=\begin{pmatrix}
	(\omega_{\mathrm{e}}+\omega^{(n-1)}-\lambda^2\mathcal{S}_f^{(n-1)}(-1))\Phie^{(n-1)}+\lambda\,\a{f}\Phig^{(n)}\\
	\omega^{(n)}\Phig^{(n)}+\lambda\frac{1}{\omega^{(n)}+1}\adag{f}\Phie^{(n-1)}
	\end{pmatrix}.
	\end{equation}
	Besides, for all $f\in\hilb_{-1}$, there exists a sequence $\{f^i\}_{i\in\mathbb{N}}\subset\hilb$ such that $H_{f^i,\omega_\e}^{(n)}\to H_{f,\omega_\e}^{(n)}$ in the norm resolvent sense.
\end{corollary}\renewcommand\arraystretch{1}
In particular, in the single-excitation sector $\hfrak^{(1)}=\mathbb{C}\oplus\hilb$,
	\begin{equation}\label{eq:singdom_1}
\mathcal{D}\left(H_{f,\omega_\e}^{(1)}\right)=\left\{
\begin{pmatrix}
a\\\xi-a\lambda\frac{1}{\omega+1}f
\end{pmatrix}:\;a\in\mathbb{C},\;\xi\in\mathcal{D}(\omega)	
\right\},
\end{equation}
and
\begin{equation}
H_{f,\omega_\e}^{(1)}\begin{pmatrix}
a\\\xi-a\lambda\frac{1}{\omega+1}f
\end{pmatrix}=\begin{pmatrix}
\left(\omega_{\mathrm{e}}-\lambda^2\mathcal{S}_f^{(0)}(-1)\right)a+\lambda\braket{f,\xi}_{-1,1}\\
\omega\xi+\lambda\frac{a}{\omega+1}f
\end{pmatrix},
\end{equation}
where (see Remark~\ref{remark:self}) the quantity $\mathcal{S}_f^{(0)}(-1)$ reduces to a number:
\begin{equation}\label{eq:shift}
\mathcal{S}_f^{(0)}(-1)=\int\frac{|f(k)|^2}{\omega(k)+1}\,\mathrm{d}\mu,
\end{equation}
yielding a shift to the excitation energy of the atom. This result is thus compatible, up to a different representation of the domain, to the one in Ref.~\cite{facchi2021spectral} for the Friedrichs-Lee model.

Finally, while not discussed here, an extension of Theorem~\ref{thm:singrwa} to the $r$-atom model, briefly introduced in Remark~\ref{remark:natom}, is, in principle, possible: while more involved, the structure of the resolvent equation is similar, and we expect the resolvent to depend by a concatenated family of propagators, whose structure is similar to the one of the single-atom propagator $\mathcal{G}_{f,\omega_\e}^{-1}(z)$. This interesting structure will be thoroughly investigated elsewhere.

\section{The rotating-wave spin-boson model with $f\in\hilb_{-2}$}\label{sec:singrwa2}

The construction of the rotating-wave spin-boson model fails to apply directly when $f\in\hilb_{-s}\setminus\hilb_{-1}$ for some $s>1$. This can be best understood by looking at the function $\mathcal{S}_f(z)$ in Eq.~\eqref{eq:gf}. As long as $f\in\hilb_{-1}$, that operator is clearly well-defined because of the properties of the singular creation and annihilation operators; this is no longer true if $f\in\hilb_{-s}\setminus\hilb_{-1}$ for any $s>1$: since $\adag{f}$ has values in $\focks_{-s}$, the operator $\mathcal{S}_f(z)$ in Eq.~\eqref{eq:sf} is unavoidably ill-defined.

The main contribution of this section will be to show that, under some additional assumptions, for $f\in\hilb_{-2}\setminus\hilb_{-1}$ (and thus $f\in\hilb_{-s}\setminus\hilb_{-1}$ for all $s\in(1,2]$), it is still possible to define a singular generalization of the rotating-wave spin-boson model. The ``price to pay'', as anticipated, is that the excitation energy $\omega_{\e}$ that appears in the formal definition of the model ``must diverge''.

We will start this last section with a heuristic discussion in Subsection~\ref{subsec:heuristics} where the basic idea at the ground of our procedure will be explained; henceforth, after providing some technical results in Subsection~\ref{subsec:renormalized}, the main results of the section (Theorem~\ref{thm:singrwa2} and Prop.~\ref{prop:fiber}) will be presented in Subsection~\ref{subsec:renormalized2}.

\subsection{Heuristics}\label{subsec:heuristics}
Given $f\in\hilb_{-s}$ for $s>1$, the quantity
\begin{equation}
	\mathcal{S}_f(z)=\a{f}\frac{1}{\dOmega-z}\adag{f},\qquad z\in\mathbb{C}\setminus[m,\infty)
\end{equation}
is, as previously remarked, ill-defined as a bounded operator on $\focks$, even if one interprets $\adag{f}$, $\a{f}$ as singular creation and annihilation operators respectively in $\mathcal{B}(\focks,\focks_{-s})$ and $\mathcal{B}(\focks_{+s},\focks)$: indeed, $(\dOmega-z)^{-1}\adag{f}$ maps continuously $\focks$ on $\focks_{-s+2}$, which, as long as $s>1$, is larger than $\focks_{+s}$.

Still, the \textit{difference} between the values of $\mathcal{S}_f(z)$ in two points of the complex plane does make sense for a larger class of form factors: given $z,z_0\in\mathbb{C}\setminus\mathbb{R}$, a simple application of the first resolvent identity yields
\begin{eqnarray}\label{eq:difference}
	\mathcal{S}_f(z)-\mathcal{S}_f(z_0)&=&\a{f}\left[\frac{1}{\dOmega-z}-\frac{1}{\dOmega-z_0}\right]\adag{f}\nonumber\\
	&=&(z-z_0)\,\a{f}\frac{1}{\dOmega-z}\frac{1}{\dOmega-z_0}\adag{f},
\end{eqnarray}
and, while the right-hand side of Eq.~\eqref{eq:difference} only makes sense when $s=1$, its left-hand side \textit{does} make sense for $f\in\hilb_{-s}$ whenever $s\leq2$. In particular, taking $z_0=\bar{z}$ above, this is true for the skew-adjoint part of $\mathcal{S}_f(z)$:
\begin{equation}
\Im\mathcal{S}_f(z)=\Im z\;\a{f}\frac{1}{\dOmega-z}\frac{1}{\dOmega-\bar{z}}\adag{f}.
\end{equation}
This simple observation suggests that it may be possible to extend the construction of the previous section for this larger class of form factors by properly redefining $\mathcal{S}_f(z)$ in such a way to make the left-hand side of Eq.~\eqref{eq:difference} well-defined without affecting the right-hand side.

For this purpose, let us have a look at the explicit expression of $\mathcal{S}_f(z)$, for $f\in\hilb_{-1}$, on $\Psi\in\focks_{+1}$, which, as discussed in Remark~\ref{remark:self}, is given by Eq.~\eqref{eq:sfn}. We are free to rewrite the latter by adding and subtracting a constant:
\begin{equation}\label{eq:addsub}
	S^{(n)}_f(z)\Psi^{(n)}=\tilde{\mathcal{S}}_f^{(n)}(z)\Psi^{(n)}-\left(\int\frac{|f(\kappa)|^2}{\omega(\kappa)}\,\mathrm{d}\mu(\kappa)\right)\Psi^{(n)},
\end{equation}
with the term between parentheses corresponding exactly to $\|f\|_{-1}^2$. Clearly, given $z,z_0\in\mathbb{C}\setminus\mathbb{R}$, the difference between the values of both operators at $z$ and $z_0$ is unchanged:
\begin{equation}
	\mathcal{S}_f(z)-\mathcal{S}_f(z_0)=\tilde{\mathcal{S}}_f(z)-\tilde{\mathcal{S}}_f(z_0).
\end{equation}
If $f\in\hilb_{-s}\setminus\hilb_{-1}$ for $s\in(1,2]$, Eq.~\eqref{eq:addsub} is ill-defined. However, we will show that, under some assumptions, in such a case $\tilde{S}_f(z)$ is still a well-defined operator on a domain containing $\mathcal{D}(\dOmega)$. Heuristically, this means that Eq.~\eqref{eq:sfn} will still define a ``good'' operator provided that we subtract an ``infinite constant'' from it. 

Consequently, looking at the definition of the operator $\tilde{\mathcal{G}}_{f,\tilde{\omega}_\e}(z)$, if the same constant if added to the excitation energy $\omega_{\e}$ of the atom, ``the two infinities will cancel out'' thus yielding a legitimate propagator.

\subsection{The renormalized propagator}\label{subsec:renormalized}
We shall now translate this idea into the mathematical realm, with the aid of the following additional assumption which will be required to prove Theorem~\ref{thm:singrwa2} (but, importantly, not for Prop.~\ref{prop:fiber}).
\begin{hypothesis}\label{hyp}
	Given $s\in[1,2]$, there exist $r\in[s-1,1]$ and $C_f>0$ such that, for all $n\in\mathbb{N}$,
	\begin{equation}\label{eq:growth}
		\int\frac{|f(k)|^2}{\left[\omega(k)+(n-1)m\right]^s}\,\mathrm{d}\mu(k)\leq\frac{C_f}{n^{s-r}}.
	\end{equation}
\end{hypothesis}
Such an assumption obviously implies $f\in\hilb_{-s}$ (just take $n=1$ in Eq.~\eqref{eq:growth}) but, moreover, imposes that such an integral must ``scale sufficiently quickly'', i.e. at least as fast as $\mathcal{O}(n^{-(s-1)})$. It will be useful to provide an explicit, and physically meaningful, example.

\begin{example}
	 Consider a one-dimensional boson field with Klein-Gordon dispersion relation, i.e. $\omega:\mathbb{R}\rightarrow\mathbb{R}$ given by $\omega(k)=\sqrt{k^2+m^2}$ for some $m>0$, and $\mu$ be the Lebesgue measure on $\mathbb{R}$. We shall set $m=1$ hereafter without loss of generality. A paradigmatic example of singular form factor (which cannot be treated with the formalism of the previous section, hence requiring renormalization), already discussed in the Introduction, is given by a \textit{flat} coupling $f(k)=\text{const.}$: in such a case,
	\begin{equation}
		\int_{-\infty}^{\infty}\frac{|f(k)|^2}{\omega(k)^s}\;\mathrm{d}k\propto\int_{-\infty}^{\infty}\frac{1}{(k^2+1)^{s/2}}\;\mathrm{d}k=\begin{cases}
			+\infty,&s\in[0,1];\\
			\frac{\sqrt{\pi}\,\Gamma \left(\frac{s-1}{2}\right)}{\Gamma \left(\frac{s}{2}\right)},&s>1,
		\end{cases}
	\end{equation}
	so that $f\in\hilb_{-s}$ for all $s>1$. In particular, taking $s=2$, we have
	\begin{eqnarray}
		\int_{-\infty}^{\infty}\frac{|f(k)|^2}{[\omega(k)+(n-1)m]^2}\;\mathrm{d}k&\propto&\int_{-\infty}^{\infty}\frac{1}{\left[\sqrt{k^2+1}+n-1\right]^2}\nonumber\\
		&\leq&\int_{-\infty}^{\infty}\frac{1}{k^2+1+(n-1)^2}\;\mathrm{d}k\nonumber\\
		&=&\frac{\pi}{\sqrt{(n-1)^2+1}}=\mathcal{O}(n^{-1}),
	\end{eqnarray}
	so that the aforementioned assumption is satisfied (for example) with $s=2$ and $r=1$.
\end{example}

\begin{lemma}\label{lemma:relbound}
	Let $z\in\mathbb{C}\setminus[m,\infty)$, suppose that $f$ satisfies Hypothesis~\ref{hyp} for some $s\in[1,2]$ and $r\in[s-1,1]$. Define the operator $\tilde{\mathcal{S}}_f(z)$ with domain $\mathcal{D}\!\left(\tilde{S}_f(z)\right)=\mathcal{D}(\dOmega^r)$ and acting as
	\small
	\begin{eqnarray}\label{eq:sfntilde}
		\left(\tilde{\mathcal{S}}_f^{(n)}(z)\Psi^{(n)}\right)(k_1,\dots,k_n)&\!=\!&\sum_{j=1}^n\left(\int\frac{\overline{f(\kappa)}\Psi^{(n)}(k_1,\dots,\overbrace{\kappa}^{j\text{th}},\dots,k_n)}{\omega(\kappa)+\sum_{\ell=1}^n\omega(k_\ell)-z}\,\mathrm{d}\mu(\kappa)\right)f(k_j)\nonumber\\&&+\!\left[\int|f(\kappa)|^2\left(\frac{1}{\omega(\kappa)+\sum_{j=1}^n\omega(k_j)-z}-\frac{1}{\omega(k)}\right)\mathrm{d}\mu(\kappa)\right]\!\Psi^{(n)}(k_1,\dots,k_n).\nonumber\\
	\end{eqnarray}\normalsize	
	Then $\tilde{\mathcal{S}}_f(z)$ is a well-defined operator satisfying the following properties:
	\begin{itemize}
		\item[(i)] if $r=1$, then $\tilde{\mathcal{S}}_f(z)$ is relatively bounded with respect to $\dOmega$;
		\item[(ii)] if $r\in[s-1,1)$, then $\tilde{\mathcal{S}}_f(z)$ is \textit{infinitesimally} relatively bounded with respect to $\dOmega$.
	\end{itemize}
	Besides, for all $\Psi\in\mathcal{D}(\dOmega^r)$ and $z,z_0\in\mathbb{C}\setminus[m,\infty)$, the following equality holds:
	\begin{eqnarray}\label{eq:diff}
		\left[\tilde{S}_f(z)-\tilde{S}_f(z_0)\right]\Psi&=&\a{f}\left[\frac{1}{\dOmega-z}-\frac{1}{\dOmega-z_0}\right]\adag{f}\Psi\nonumber\\
		&=&(z-z_0)\a{f}\frac{1}{\dOmega-z}\frac{1}{\dOmega-z_0}\adag{f}\Psi
	\end{eqnarray}
	and, if $f\in\hilb_{-1}$,
	\begin{equation}\label{eq:diff2}
		\mathcal{S}_f(z)\Psi=\tilde{\mathcal{S}}_f(z)\Psi-\|f\|^2_{-1}\Psi.
	\end{equation}
\end{lemma}
\begin{proof}
	It will suffice to verify the claim for $z=0$. For definiteness, let us write $\tilde{S}_f^{(n)}(0)=\mathcal{R}^{(n)}_f+\mathcal{T}^{(n)}_f$, with
	\small\begin{eqnarray}\label{eq:rfn}
		\left(\mathcal{R}_f^{(n)}\Psi^{(n)}\right)(k_1,\dots,k_n)&=&\sum_{j=1}^n\left(\int\frac{\overline{f(\kappa)}\Psi^{(n)}(k_1,\dots,\overbrace{\kappa}^{j\text{th}},\dots,k_n)}{\omega(\kappa)+\sum_{\ell=1}^n\omega(k_\ell)}\,\mathrm{d}\mu(\kappa)\right)f(k_j);\\\label{eq:tfn}
		\left(\mathcal{T}^{(n)}_f\Psi^{(n)}\right)(k_1,\dots,k_n)&=&\left[\int|f(\kappa)|^2\left(\frac{1}{\omega(\kappa)+\sum_{j=1}^n\omega(k_j)}-\frac{1}{\omega(\kappa)}\right)\mathrm{d}\mu(\kappa)\right]\!\Psi^{(n)}(k_1,\dots,k_n)\nonumber\\
		&=&\left[\int\frac{|f(\kappa)|^2}{\omega(\kappa)\left[\omega(\kappa)+\sum_{j=1}^n\omega(k_j)\right]}\mathrm{d}\mu(\kappa)\right]\!\left[-\sum_{j=1}^n\omega(k_j)\right]\!\Psi^{(n)}(k_1,\dots,k_n).\nonumber\\
	\end{eqnarray}\normalsize
	We shall verify that
	\begin{itemize}
		\item if Hypothesis~\ref{hyp} holds, $\mathcal{R}_f$, with domain $\mathcal{D}(\dOmega^r)$, is a relatively bounded operator with respect to $\dOmega^r$;
		\item $\mathcal{T}_f$, with domain $\mathcal{D}(\dOmega^{s-1})$, is a relatively bounded operator with respect to $\dOmega^{s-1}$, and thus \textit{a fortiori} (since $r\geq s-1$) to $\dOmega^r$;
	\end{itemize}
this will prove $(i)$ and $(ii)$, since, for $r<1$, $\dOmega^r$ is infinitesimally $\dOmega$-bounded.
	
	Let us start from the second one. Given $\Psi\in\mathcal{D}\!\left(\dOmega^{s-1}\right)$, we have\footnotesize
	\begin{eqnarray}
		&&\left|\left(\mathcal{T}_f^{(n)}\Psi^{(n)}\right)(k_1,\dots,k_n)\right|^2\nonumber\\
		&=&\left|\int\frac{|f(\kappa)|^2}{\omega(\kappa)\left[\omega(\kappa)+\sum_{j=1}^n\omega(k_j)\right]}\mathrm{d}\mu(\kappa)\right|^2\left[\sum_{\ell=1}^n\omega(k_\ell)\right]^2\left|\Psi^{(n)}(k_1,\dots,k_n)\right|^2\nonumber\\
		&=&\left|\int\frac{|f(\kappa)|^2}{\omega(\kappa)\left[\omega(\kappa)+\sum_{j=1}^n\omega(k_j)\right]^{s-1}\left[\omega(\kappa)+\sum_{j=1}^n\omega(k_j)\right]^{2-s}}\mathrm{d}\mu(\kappa)\right|^2\left[\sum_{\ell=1}^n\omega(k_\ell)\right]^2\left|\Psi^{(n)}(k_1,\dots,k_n)\right|^2\nonumber\\
		&\leq&\left|\int\frac{|f(\kappa)|^2}{\omega(\kappa)\left[\omega(\kappa)+\sum_{j=1}^n\omega(k_j)\right]^{s-1}\left[\sum_{j=1}^n\omega(k_j)\right]^{2-s}}\mathrm{d}\mu(\kappa)\right|^2\left[\sum_{\ell=1}^n\omega(k_\ell)\right]^2\left|\Psi^{(n)}(k_1,\dots,k_n)\right|^2\nonumber\\
		&=&\left|\int\frac{|f(\kappa)|^2}{\omega(\kappa)\left[\omega(\kappa)+\sum_{j=1}^n\omega(k_j)\right]^{s-1}}\mathrm{d}\mu(\kappa)\right|^2\left[\sum_{\ell=1}^n\omega(k_\ell)\right]^{2s-2}\left|\Psi^{(n)}(k_1,\dots,k_n)\right|^2\nonumber\\
		&\leq&\|f\|^2_{-s}\left[\sum_{\ell=1}^n\omega(k_\ell)\right]^{2s-2}\left|\Psi^{(n)}(k_1,\dots,k_n)\right|^2,
	\end{eqnarray}\normalsize
	whence, integrating on $k_1,\dots,k_n$ and summing on $n$,
	\begin{equation}
		\left\|\mathcal{T}_f\Psi\right\|_\focks\leq\|f\|_{-s}\left\|\dOmega^{s-1}\Psi\right\|_\focks.
	\end{equation}
	As for the first term, we have\footnotesize
	\begin{eqnarray}
		\left|\left(\mathcal{R}_f^{(n)}\Psi^{(n)}\right)(k_1,\dots,k_n)\right|^2&=&\left|\sum_{j=1}^n\left(\int\frac{\overline{f(\kappa)}\Psi^{(n)}(k_1,\dots,\overbrace{\kappa}^{j\text{th}},\dots,k_n)}{\omega(\kappa)+\sum_{\ell=1}^n\omega(k_\ell)}\,\mathrm{d}\mu(\kappa)\right)f(k_j)\right|^2\nonumber\\&\leq&n\sum_{j=1}^n|f(k_j)|^2\left|\int\frac{\overline{f(\kappa)}\Psi^{(n)}(k_1,\dots,\overbrace{\kappa}^{j\text{th}},\dots,k_n)}{\omega(\kappa)+\sum_{\ell=1}^n\omega(k_\ell)}\,\mathrm{d}\mu(\kappa)\right|^2\nonumber\\
		&=&n\sum_{j=1}^n|f(k_j)|^2\left|\int\frac{\overline{f(\kappa)}\Psi^{(n)}(k_1,\dots,\overbrace{\kappa}^{j\text{th}},\dots,k_n)}{\omega(\kappa)+\omega(k_j)+\sum_{\ell\neq j}^n\omega(k_\ell)}\frac{\left[\omega(\kappa)+\sum_{\ell\neq j}^n\omega(k_\ell)\right]^r}{\left[\omega(\kappa)+\sum_{\ell\neq j}^n\omega(k_\ell)\right]^r}\,\mathrm{d}\mu(\kappa)\right|^2\nonumber\\
		&\leq&n\sum_{j=1}^n\int\mathrm{d}\mu(\kappa)\:\frac{|f(k_j)|^2|f(\kappa)|^2}{\left[\omega(\kappa)+\omega(k_j)+\sum_{\ell\neq j}^n\omega(k_\ell)\right]^2\left[\omega(\kappa)+\sum_{\ell\neq j}^n\omega(k_\ell)\right]^{2r}}\times\nonumber\\
		&&\times\int\mathrm{d}\mu(\kappa')\;\left[\omega(\kappa')+\sum_{\ell\neq j}^n\omega(k_\ell)\right]^{2r}\left|\Psi^{(n)}(k_1,\dots,\overbrace{\kappa'}^{j\text{th}},\dots,k_n)\right|^2\nonumber\\
		&\leq&n\sum_{j=1}^n\int\mathrm{d}\mu(\kappa)\:\frac{|f(k_j)|^2|f(\kappa)|^2}{\left[\omega(\kappa)+\omega(k_j)+(n-1) m)\right]^2\left[\omega(\kappa)+(n-1)m\right]^{2r}}\times\nonumber\\
		&&\times\int\mathrm{d}\mu(\kappa')\;\left[\omega(\kappa')+\sum_{\ell\neq j}^n\omega(k_\ell)\right]^{2r}\left|\Psi^{(n)}(k_1,\dots,\overbrace{\kappa'}^{j\text{th}},\dots,k_n)\right|^2\nonumber\\
		&\leq&n\sum_{j=1}^n\int\mathrm{d}\mu(\kappa)\:\frac{|f(k_j)|^2|f(\kappa)|^2}{\left[\omega(k_j)+(n-1) m\right]^s\left[\omega(\kappa)+(n-1) m\right]^{2-s+2r}}\times\nonumber\\
		&&\times\int\mathrm{d}\mu(\kappa')\;\left[\omega(\kappa')+\sum_{\ell\neq j}^n\omega(k_\ell)\right]^{2r}\left|\Psi^{(n)}(k_1,\dots,\overbrace{\kappa'}^{j\text{th}},\dots,k_n)\right|^2.
	\end{eqnarray}
	\normalsize
	Since $r\geq s-1$, we have $2-s+2r\geq s$, whence\footnotesize
	\begin{eqnarray}
		\left|\left(\mathcal{R}_f^{(n)}\Psi^{(n)}\right)(k_1,\dots,k_n)\right|^2&\leq&\frac{1}{m^{2-2s+2r}}n^{2s-2r-1}\sum_{j=1}^n\int\mathrm{d}\mu(\kappa)\:\frac{|f(k_j)|^2|f(\kappa)|^2}{\left[\omega(k_j)+(n-1) m\right]^s\left[\omega(\kappa)+(n-1) m\right]^{s}}\times\nonumber\\
		&&\times\int\mathrm{d}\mu(\kappa')\;\left[\omega(\kappa')+\sum_{\ell\neq j}^n\omega(k_\ell)\right]^{2r}\left|\Psi^{(n)}(k_1,\dots,\overbrace{\kappa'}^{j\text{th}},\dots,k_n)\right|^2.
	\end{eqnarray}\normalsize
	Integrating on $k_1,\dots,k_n$, taking into account that all $n$ terms in the sum yield equal contributions, properly renaming the integration variables, and finally using Eq.~\eqref{eq:growth}, one gets\footnotesize
	\begin{eqnarray}\label{eq:play}
		\int\left|\left(\mathcal{R}_f^{(n)}\Psi^{(n)}\right)(k_1,\dots,k_n)\right|^2\:\mathrm{d}^n\mu&\leq&\frac{1}{m^{2-2s+2r}}n^{2(s-r)}\left[\int\mathrm{d}\mu(\kappa)\:\frac{|f(\kappa)|^2}{\left[\omega(\kappa)+(n-1) m\right]^{s}}\right]^2\left\|\left[\omega^{(n)}\right]^{r}\Psi^{(n)}\right\|^2_{\hilbn}\nonumber\\
		&\leq&\frac{C_f^2}{m^{2-2s+2r}}\,\left\|\left[\omega^{(n)}\right]^{r}\Psi^{(n)}\right\|^2_{\hilbn},
	\end{eqnarray}\normalsize
	whence
	\begin{equation}
		\left\|\mathcal{R}_f\Psi\right\|_\focks\leq\frac{C_f}{m^{1-s+r}}\left\|\dOmega^r\Psi\right\|_\focks,
	\end{equation}
	thus completing the proof of (i) and (ii). The final claims (Eqs.~\eqref{eq:diff}-\eqref{eq:diff2}) are an immediate consequence of the very definition of $\tilde{\mathcal{S}}_f(z)$.
\end{proof}

\begin{remark}\label{remark:regardless}
	Hypothesis~\ref{hyp} comes into play specifically in the last step of the inequality~\eqref{eq:play}, where it is used to get rid of the term $n^{s-r}$ and therefore making the action $\tilde{\mathcal{S}}_f(z)=\bigoplus_n\tilde{\mathcal{S}}^{(n)}_f(z)$ well-defined on the Fock space $\focks$. However, the same computations in the proof of Lemma~\ref{lemma:relbound} show that, \textit{regardless} of whether Hypothesis~\ref{hyp} holds, the following claim also hold: for all $n\in\mathbb{N}$, $\tilde{\mathcal{S}}_f(z)$ is relatively bounded with respect to $(\omega^{(n)})^{s-1}$. This is simply obtained by repeating the same computations as above with $r=s-1$. In particular:
	\begin{itemize}
		\item if $s=2$, then $\tilde{\mathcal{S}}^{(n)}_f(z)$ is relatively bounded with respect to $\omega^{(n)}$;
		\item if $s<2$, then $\tilde{\mathcal{S}}^{(n)}_f(z)$ is \textit{infinitesimally} relatively bounded with respect to $\omega^{(n)}$.
	\end{itemize}
This simple observation will be crucial to prove Prop.~\ref{prop:fiber}.
\end{remark}
	
	\begin{remark}\label{remark:self2}
	Reprising the discussion in Remark~\ref{remark:self}, notice that, for $n=0$, we simply have
	\begin{equation}\label{eq:self2}
		\tilde{\mathcal{S}}_f^{(0)}(z)\equiv\tilde{\Sigma}_f(z)=\int|f(k)|^2\left(\frac{1}{\omega(k)-z}-\frac{1}{\omega(k)}\right)\,\mathrm{d}\mu(k),
	\end{equation}
	which is analogous to the \textit{renormalized} self-energy of the Friedrichs-Lee model with form factor $f\in\hilb_{-2}$ studied in~\cite{facchi2021spectral}.
\end{remark}

By taking into account Lemma~\ref{lemma:relbound}, we can now state the counterpart of Lemma~\ref{lemma2} for the renormalized propagator.
\begin{lemma}\label{lemma:inv}
	Let $s\in[1,2]$, $f\in\hilb_{-s}$, and $\tilde{\omega}_\e\in\mathbb{R}$. Given $z\in\mathbb{C}\setminus\mathbb{R}$ and $\lambda\in\mathbb{R}$, define the operator $\tilde{\mathcal{G}}_{f,\tilde{\omega}_\e}(z)$ via
	\begin{equation}\label{eq:gf2}
		\tilde{\mathcal{G}}_{f,\tilde{\omega}_\e}(z)=\tilde{\omega}_\e-z+\dOmega-\lambda^2\tilde{\mathcal{S}}_f(z),
	\end{equation}
	with domain $\mathcal{D}\!\left(\tilde{\mathcal{G}}_{f,\tilde{\omega}_\e}(z)\right)=\mathcal{D}(\dOmega)$. The following facts hold:
	\begin{itemize}
		\item[(i)] if $f$ satisfies Hypothesis~\ref{hyp} with $r\in[s-1,1)$, then for all $\lambda$ the operator $\tilde{\mathcal{G}}_f(z)$ is closed, satisfies $\tilde{\mathcal{G}}_{f,\tilde{\omega}_\e}(\bar{z})=\tilde{\mathcal{G}}_{f,\tilde{\omega}_\e}(z)^*$, and admits a bounded inverse $\tilde{\mathcal{G}}_{f,\tilde{\omega}_\e}^{-1}(z)$.
		\item[(ii)] if $f$ satisfies Hypothesis~\ref{hyp} with $r=1$, then for sufficiently small $\lambda$ the operator $\tilde{\mathcal{G}}_{f,\tilde{\omega}_\e}(z)$ is closed, satisfies $\tilde{\mathcal{G}}_{f,\tilde{\omega}_\e}(\bar{z})=\tilde{\mathcal{G}}_{f,\tilde{\omega}_\e}(z)^*$, and admits a bounded inverse $\tilde{\mathcal{G}}_{f,\tilde{\omega}_\e}^{-1}(z)$;
	\end{itemize}
	in both cases, the operator norm of $\tilde{\mathcal{G}}^{-1}_{f,\tilde{\omega}_\e}(z)$ satisfies
	\begin{equation}\label{eq:boundednorm}
		\left\|\tilde{\mathcal{G}}^{-1}_{f,\tilde{\omega}_\e}(z)\right\|_{\mathcal{B}(\focks)}\leq\frac{1}{|\Im z|}.
	\end{equation}
	Finally, if $s=1$, then
	\begin{equation}
		\tilde{\mathcal{G}}_{f,\tilde{\omega}_\e}(z)=\mathcal{G}_{f,\omega_\e}(z),\quad\text{where}\quad\omega_\e=\tilde{\omega}_\e-\|f\|^2_{-1}.
	\end{equation}
\end{lemma}
\begin{proof}
	First of all, notice that the operators $\tilde{\mathcal{S}}_f(z)$ and $\tilde{\mathcal{S}}_f(\bar{z})$ have the same domain $\mathcal{D}\!\left(\tilde{S}_f(z)\right)=\mathcal{D}\!\left(\tilde{S}_f(\bar{z})\right)=\mathcal{D}(\dOmega^r)$, and satisfy
	\begin{equation}
		\Braket{\Psi,\tilde{\mathcal{S}}_f(z)\Phi}_\focks=\Braket{\tilde{\mathcal{S}}_f(\bar{z})\Psi,\Phi}_\focks
	\end{equation}
	for all $\Psi,\Phi\in\mathcal{D}\!\left(\dOmega^r\right)$, whence $\mathcal{S}_f(\bar{z})\subseteq\mathcal{S}_f(z)^*$; in particular, the actions of the two latter operators coincide on $\mathcal{D}(\dOmega)$. This also readily implies that, for all $\Psi\in\mathcal{D}\!\left(\dOmega^r\right)$,
	\begin{equation}\label{eq:impart}
		\Im\Braket{\Psi,\tilde{\mathcal{S}}_f(z)\Psi}_\focks\geq0.
	\end{equation}	
	Now, the following equality:
	\begin{eqnarray}
		\tilde{\mathcal{G}}_{f,\tilde{\omega}_\e}(z)^*&=&\left[\tilde{\omega}_\e-z+\dOmega-\lambda^2\tilde{\mathcal{S}}_f(z)\right]^*=	\tilde{\omega}_\e-z+\dOmega-\lambda^2\tilde{\mathcal{S}}_f(z)^*\nonumber\\&=&\tilde{\omega}_\e-z+\dOmega-\lambda^2\tilde{\mathcal{S}}_f(\bar{z})=\tilde{\mathcal{G}}_{f,\tilde{\omega}_\e}(\bar{z}),
	\end{eqnarray}
	holds, also with $\tilde{\mathcal{G}}_{f,\tilde{\omega}_\e}(z)$ being a closed operator, whenever both $\lambda^2\tilde{S}_f(z)$ and $\lambda^2\tilde{S}_f(z)^*$ are relatively bounded with respect to $\dOmega$ with relative bound smaller than one; see e.g.~\cite[Corollary 1]{hess1970perturbation}. By Lemma~\ref{lemma:relbound}, this indeed holds either for every value of $\lambda$ if $r\in[s-1,1)$ (since, in such a case, the relative bound is zero), or for sufficiently small values of $\lambda$ if $r=1$ (since, in such a case, the relative bound is finite).
	
	By the equality $\tilde{\mathcal{G}}_{f,\tilde{\omega}_\e}(z)^*=\tilde{\mathcal{G}}_{f,\tilde{\omega}_\e}(\bar{z})$ and the inequality~\eqref{eq:impart}, the existence of a bounded inverse $\tilde{\mathcal{G}}_{f,\tilde{\omega}_\e}(z)^{-1}$ with operator norm satisfying Eq.~\eqref{eq:boundednorm} follows from the same argument as in the proof of Lemma~\ref{lemma:inv}. Finally, the last claim is an immediate consequence of the last claim of Lemma~\ref{lemma:relbound}.
\end{proof}

\subsection{Extension of the model to form factors $f\in\hilb_{-2}$}\label{subsec:renormalized2}
We can finally state the counterpart of Theorem~\ref{thm:singrwa} for form factors $f\in\hilb_{-s}\setminus\hilb_{-1}$ for $s\in[1,2]$. 
\begin{theorem}\label{thm:singrwa2}
	Given $s\in[1,2]$, let $f\in\hilb_{-s}$ satisfying Hypothesis~\ref{hyp} for some $r\in[s-1,1]$. Let $\tilde{H}_{f,\tilde{\omega}_\e}$ be the operator on $\focks$ with domain
	\begin{equation}\label{eq:domrwa2}
		\mathcal{D}\!\left(\tilde{H}_{f,\tilde{\omega}_\e}\right)=\left\{\begin{pmatrix}
			\Phie\\\Phi_g-\lambda\frac{1}{\dOmega+1}\adag{f}\Phie
		\end{pmatrix}:\;\Phie,\Phig\in\mathcal{D}\!\left(\dOmega\right)\right\}
	\end{equation}
	acting as\renewcommand\arraystretch{1.5}
	\begin{equation}\label{eq:action2}
		\tilde{H}_{f,\tilde{\omega}_\e}\begin{pmatrix}
			\Phie\\\Phig-\lambda\frac{1}{\dOmega+1}\adag{f}\Phie
		\end{pmatrix}=\begin{pmatrix}
			\left(\tilde{\omega}_{\mathrm{e}}+\dOmega-\lambda^2\tilde{S}_f(-1)\right)\Phie+\lambda\a{f}\Phig\\
			\dOmega\Phig+\lambda\frac{1}{\dOmega+1}\adag{f}\Phie
		\end{pmatrix}.
	\end{equation}
	Then the following facts hold either for all values of $\lambda$ (if $r<1$) or for sufficiently small $\lambda$ (if $r=1$):
	\begin{enumerate}
		\item [(i)] for $f\in\hilb_{-1}$, the following equality holds: 
		\begin{equation}\label{eq:renormalizat}
			\tilde{H}_{f,\tilde{\omega}_\e}=H_{f,\omega_\e},\qquad\text{where}\qquad\omega_\e=\tilde{\omega}_\e-\|f\|^2_{-1},
		\end{equation}
		with $H_{f,\omega_\e}$ being the rotating-wave spin-boson model with form factor $f\in\hilb_{-1}$ and excitation energy $\omega_\e$, cf. Theorem~\ref{thm:singrwa};
		\item [(ii)] for $f\in\hilb_{-s}$, $\tilde{H}_{f,\tilde{\omega}_\e}$ is a self-adjoint operator on $\focks$, whose resolvent reads, for all $z\in\mathbb{C}\setminus\mathbb{R}$,
		\begin{equation}\label{eq:ressing2}
			\frac{1}{\tilde{H}_{f,\tilde{\omega}_\e}-z}\begin{pmatrix}
				\Psie\\\Psig
			\end{pmatrix}=\begin{pmatrix}
				\tilde{\mathcal{G}}_{f,\tilde{\omega}_\e}^{-1}(z)\Bigl(\Psie-\lambda\a{f}\frac{1}{\dOmega-z}\Psig\Bigr)\\
				\frac{1}{\dOmega-z}\Psig-\lambda\frac{1}{\dOmega-z}\adag{f}\tilde{\mathcal{G}}^{-1}_f(z)\Bigl(\Psie-\lambda\a{f}\frac{1}{\dOmega-z}\Psig\Bigr)
			\end{pmatrix},
		\end{equation}
		with $\tilde{\mathcal{G}}_{f,\tilde{\omega}_\e}(z)$ as defined in Eq.~\eqref{eq:gf2};
		\item [(iii)] given $f\in\hilb_{-s}$, there is a sequence $\{f^i\}_{i\in\mathbb{N}}\subset\hilb$ of normalizable form factors and a sequence $\{\omega_\e^i\}_{i\in\mathbb{N}}\subset\mathbb{R}$ such that $H_{f^i,\omega_\e^i}\to \tilde{H}_{f,\tilde{\omega}_\e}$ in the strong resolvent sense;
		\item [(iv)] conversely, given any $\{f^i\}_{i\in\mathbb{N}}\subset\hilb$, $f\in\hilb_{-s}\setminus\hilb_{-1}$ such that
		\begin{equation}\label{eq:renormalization}
			\|f^i-f\|_{-s}\to0\qquad\text{and}\qquad \lim_{i\to\infty}\left(\omega_\e^i+\|f^i\|_{-1}^2\right)=:\tilde{\omega}_\e\in\mathbb{R},
		\end{equation}
		then the sequence of regular rotating-wave spin-boson models $\left\{H_{f^i,\omega_\e^i}\right\}_{i\in\mathbb{N}}$ converges to the operator $\tilde{H}_{f,\tilde{\omega}_\e}$ in the strong resolvent sense.
	\end{enumerate}
\end{theorem}
\begin{proof}
	$(i)$ Given $f\in\hilb_{-1}$, the claim follows directly from the equality $\tilde{\mathcal{S}}_{f}(z)=\mathcal{S}_f(z)$ provided by Lemma~\ref{lemma:relbound}.
	
	$(ii)$ Consider the equation\renewcommand\arraystretch{1}
	\begin{equation}
		\left(\tilde{H}_{f,\tilde{\omega}_\e}-z\right)\begin{pmatrix}
			\Phie\\\Phig-\lambda\frac{1}{\dOmega+1}\adag{f}\Phie
		\end{pmatrix}=\begin{pmatrix}
			\Psie\\\Psig
		\end{pmatrix}
	\end{equation}
	for $z\in\mathbb{C}\setminus\mathbb{R}$, that is,\renewcommand\arraystretch{1.5}
	\begin{equation}
		\begin{pmatrix}	\left(\tilde{\omega}_{\mathrm{e}}+\dOmega-z-\lambda^2\tilde{S}_f(-1)\right)\Phie+\lambda\a{f}\Phig\\
			\left(\dOmega-z\right)\Phig+(z+1)\lambda\frac{1}{\dOmega+1}\adag{f}\Phie
		\end{pmatrix}=\begin{pmatrix}
			\Psie\\\Psig
		\end{pmatrix}.
	\end{equation}\renewcommand\arraystretch{1}
	The second equation yields
	\begin{eqnarray}\label{eq:phig2}
		\Phig=\frac{1}{\dOmega-z}\Psig-\lambda(z+1)\frac{1}{\dOmega-z}\frac{1}{\dOmega+1}\adag{f}\Phie;
	\end{eqnarray}
	Substituting into the first one and taking into account Eq.~\eqref{eq:diff} in Lemma~\ref{lemma:relbound}, we get the following equation
	\begin{equation}
		\tilde{\mathcal{G}}_{f,\tilde{\omega}_\e}(z)\Phie=\Psie-\lambda\,\a{f}\frac{1}{\dOmega-z}\Psig,
	\end{equation}
	which, by Lemma~\ref{lemma:inv}, can be uniquely solved via $\tilde{\mathcal{G}}^{-1}_{f,\tilde{\omega}_\e}(z)$. Analogous calculations as the ones in the proof of Theorem~\ref{thm:singrwa}(ii) finally yield the desired result.
	
	$(iii)$ Let $f\in\hilb_{-s}$. Since $\hilb$ is densely embedded into $\hilb_{-s}$, there exists a sequence $\{f^i\}_{i\in\mathbb{N}}\subset\hilb$ such that $\|f^i-f\|_{-s}\to0$ as $i\to\infty$. Setting $\omega_\e^i=\tilde{\omega}_\e-\|f^i\|^2_{-1}$, clearly
	\begin{equation}
		\mathcal{G}_{f^i,\omega_\e^i}(z)=\tilde{\mathcal{G}}_{f^i,\tilde{\omega}_\e}(z)
	\end{equation}
	and it converges strongly to $\tilde{\mathcal{G}}_{f^i,\tilde{\omega}_\e^i}(z)$ as $i\to\infty$. Now, we have
	\begin{equation}
		\tilde{\mathcal{G}}^{-1}_{f,\tilde{\omega}_\e}(z)-\tilde{\mathcal{G}}^{-1}_{f^i,\tilde{\omega}_\e}(z)=	\tilde{\mathcal{G}}^{-1}_{f^i,\tilde{\omega}_\e}(z)\left[\tilde{\mathcal{G}}_{f^i,\tilde{\omega}_\e}(z)-\tilde{\mathcal{G}}_{f,\tilde{\omega}_\e}(z)\right]\tilde{\mathcal{G}}^{-1}_{f,\tilde{\omega}_\e}(z).
	\end{equation}
	Since $\tilde{\mathcal{G}}^{-1}_{f,\tilde{\omega}_\e}(z)$ has a bounded inverse, the most general vector $\Psi\in\focks$ can be written as $\Psi=\tilde{\mathcal{G}}^{-1}_{f,\tilde{\omega}_\e}(z)\Phi$ for some $\Phi\in\mathcal{D}(\dOmega)$. Consequently,
	\begin{eqnarray}
		\left\|\left[\tilde{\mathcal{G}}^{-1}_{f^i,\tilde{\omega}_\e}(z)-\tilde{\mathcal{G}}^{-1}_{f,\tilde{\omega}_\e}(z)\right]\Psi\right\|_\focks&=&	\left\|\tilde{\mathcal{G}}^{-1}_{f^i,\tilde{\omega}_\e}(z)\left[\tilde{\mathcal{G}}_{f^i,\tilde{\omega}_\e}(z)-\tilde{\mathcal{G}}_{f,\tilde{\omega}_\e}(z)\right]\tilde{\mathcal{G}}^{-1}_{f,\tilde{\omega}_\e}(z)\Psi\right\|_\focks\nonumber\\
		&=&\left\|\tilde{\mathcal{G}}^{-1}_{f^i,\tilde{\omega}_\e}(z)\left[\tilde{\mathcal{G}}_{f^i,\tilde{\omega}_\e}(z)-\tilde{\mathcal{G}}_{f,\tilde{\omega}_\e}(z)\right]\Phi\right\|_\focks\nonumber\\&\leq&\left\|\tilde{\mathcal{G}}^{-1}_{f^i,\tilde{\omega}_\e}(z)\right\|_{\mathcal{B}(\focks)}\left\|\left[\tilde{\mathcal{G}}_{f^i,\tilde{\omega}_\e}(z)-\tilde{\mathcal{G}}_{f,\tilde{\omega}_\e}(z)\right]\Phi\right\|_\focks\nonumber\\
		&\leq&\frac{1}{|\Im z|}\left\|\left[\tilde{\mathcal{G}}_{f^i,\tilde{\omega}_\e}(z)-\tilde{\mathcal{G}}_{f,\tilde{\omega}_\e}(z)\right]\Phi\right\|_\focks\to0,
	\end{eqnarray}
	which implies the claim.  $(iv)$ is proven analogously. 
\end{proof}
Theorem~\ref{thm:singrwa2} represents the desired generalization of Theorem~\ref{thm:singrwa} to the case of singular form factors $f\in\hilb_{-s}$, $s>1$. The interpretation of $\tilde{\omega}_\e$ as a \textit{renormalized} excitation energy of the spin, as opposed to the bare excitation energy $\omega_\e$ of the original model, should now be clear: whenever $f\notin\hilb_{-1}$, any approximating sequence of regularized form factors $\{f^i\}_{i\in\mathbb{N}}$ yields a diverging quantity $\|f^i\|^2_{-1}$, so that (see Eq.~\eqref{eq:renormalization}) any sequence of approximating spin-boson models must be characterized by an excitation energy $\omega^i_\e$ diverging as well, as opposed to the case considered in Theorem~\ref{thm:singrwa} where no renormalization of the excitation energy was required.

\begin{remark}\label{rem:vacuum2}
Reprising the discussion in Remark~\ref{rem:vacuum}, notice that the state $\Psi_0$ in Eq.~\eqref{eq:vacuum}, corresponding to the excited state of the atom coupled with the vacuum of the boson field, satisfies $\Psi_0\notin\mathcal{Q}(\tilde{H}_{f,\tilde{\omega}_\e})$ (the latter being the form domain of $\tilde{H}_{f,\tilde{\omega}_\e}$) whenever $f\in\hilb_{-2}\setminus\hilb_{-1}$. Physically, this means that the total energy distribution of such a state has an infinite average as well. Summing up, in the rotating-wave spin-boson model, the total energy distribution associated with the state $\Psi_0$ has
\begin{itemize}
	\item if $f\in\hilb$, finite average and finite variance;
	\item if $f\in\hilb_{-1}\setminus\hilb$, finite average but infinite variance;
	\item if $f\in\hilb_{-2}\setminus\hilb_{-1}$, infinite average and infinite variance.
\end{itemize}
This discussion mirrors the situation already observed in \cite{facchi2021spectral} for the single-excitation sector $\tilde{H}_{f,\tilde{\omega}_\e}$ of the model.
\end{remark}

To conclude, let us discuss the restrictions $\tilde{H}^{(n)}_{f,\tilde{\omega}_\e}$ of the model to the finite-excitation sectors $\hfrak^{(n)}$ introduced in Section~\ref{sec:singrwa}, thus finding the counterpart of Corollary~\ref{coroll1} for form factors $f\in\hilb_{-s}\setminus\hilb_{-1}$ for $s\in[1,2]$. While Hypothesis~\ref{hyp} is required for the self-adjointness of the model on the full Fock space, such an assumption can be relaxed when taking into account each $\tilde{H}^{(n)}_{f,\tilde{\omega}_\e}$ separately:
\begin{proposition}\label{prop:fiber}
	Let $f\in\hilb_{-s}$, $s\in[1,2]$. For all $n\in\mathbb{N}$, let $\tilde{H}_{f,\tilde{\omega}_\e}^{(n)}$ be the operator on $\hfrak^{(n)}$ with domain
	\begin{equation}\label{eq:singdom_n2}
		\mathcal{D}\left(\tilde{H}_{f,\tilde{\omega}_\e}^{(n)}\right)=\left\{
		\begin{pmatrix}
			\Phie^{(n-1)}\\\Phig^{(n)}-\lambda\frac{1}{\omega^{(n)}+1}\adag{f}\Phie^{(n-1)}
		\end{pmatrix}:\;\Phie^{(n-1)}\in\mathcal{D}(\omega^{(n-1)}),\;\Phig^{(n)}\in\mathcal{D}(\omega^{(n)}),	
		\right\},
	\end{equation}
	acting as
	\begin{equation}
		\tilde{H}_{f,\omega_\e}^{(n)}\begin{pmatrix}
			\Phie^{(n-1)}\\\Phig^{(n)}-\lambda\frac{1}{\omega^{(n)}+1}\adag{f}\Phie^{(n-1)}
		\end{pmatrix}=\begin{pmatrix}
			(\tilde{\omega}_{\mathrm{e}}+\omega^{(n-1)}-\lambda^2\tilde{\mathcal{S}}_f^{(n-1)}(-1))\Phie^{(n-1)}+\lambda\,\a{f}\Phig^{(n)}\\
			\omega^{(n)}\Phig^{(n)}+\lambda\frac{1}{\omega^{(n)}+1}\adag{f}\Phie^{(n-1)}
		\end{pmatrix}.
	\end{equation}
Then the following facts hold either for all values of $\lambda$ (if $s<2$) or for sufficiently small $\lambda$ (if $s=2$):
\begin{itemize}
	\item[(i)] $\tilde{H}^{(n)}_{f,\tilde{\omega}_\e}$ is a self-adjoint operator on $\hfrak^{(n)}$;
	\item[(ii)] there exists sequences $\{f^i\}_{i\in\mathbb{N}}\subset\hilb$ and $\{\omega^i_\e\}_{i\in\mathbb{N}}\subset\mathbb{R}$ such that $H_{f^i,\omega^i_\e}^{(n)}\to \tilde{H}_{f,\tilde{\omega}_\e}^{(n)}$ in the strong resolvent sense,
\end{itemize}
and the same holds for the operator $\bigoplus_{n\in\mathbb{N}}\tilde{H}^{(n)}_{f,\tilde{\omega}_\e}$ on $\focks$.

Finally, if $f$ satisfies Hypothesis~\ref{hyp} for some $r\in[s-1,1]$, then
\begin{equation}\label{eq:directsum}
	\tilde{H}_{f,\tilde{\omega}_\e}=\bigoplus_{n\in\mathbb{N}}\tilde{H}^{(n)}_{f,\tilde{\omega}_\e}
\end{equation}
either for all values of $\lambda$ (if $r<1$) or for sufficiently small $\lambda$ (if $r=1$).
\end{proposition}\renewcommand\arraystretch{1}
\begin{proof}
	Lemma~\ref{lemma:relbound} implies (see Remark~\ref{remark:regardless}) that, for each fixed $n\in\mathbb{N}$ and $z\in\mathbb{C}\setminus[m,\infty)$, the operator $\tilde{S}^{(n)}_f(z)$ is relatively bounded, either with zero bound (if $s<2$) or finite bound (if $s=2$), with respect to $\omega^{(n)}$; an argument analogous to the one in Lemma~\ref{lemma:inv} shows that, in either case, $\tilde{\mathcal{G}}^{(n)}_f(z)$ admits a bounded inverse. Following the same arguments as in the proof of Theorem~\ref{thm:singrwa2}, the claims $(i)$ and $(ii)$ are shown. The properties of the direct sum of self-adjoint operators on Hilbert spaces readily implies that the same properties hold for $\bigoplus_{n\in\mathbb{N}}\tilde{H}^{(n)}_{f,\tilde{\omega}_\e}$.
	
	The final claim simply follows by noticing that, whenever the assumptions of Theorem~\ref{thm:singrwa2} are satisfied, the restriction of $\tilde{H}_{f,\tilde{\omega}_\e}$ to $\hfrak^{(n)}$ coincides with $\tilde{H}^{(n)}_{f,\tilde{\omega}_\e}$.
\end{proof}
\begin{remark}
Notice that Eq.~\eqref{eq:directsum} \textit{always} defines a self-adjoint operator on $\fock$, its domain being
\begin{equation}
	\mathcal{D}\!\left(\bigoplus_{n\in\mathbb{N}}\tilde{H}^{(n)}_{f,\tilde{\omega}_\e}\right)=\bigoplus_{n\in\mathbb{N}}\mathcal{D}\!\left(\tilde{H}^{(n)}_{f,\tilde{\omega}_\e}\right),
\end{equation}
even when $\tilde{H}_{f,\tilde{\omega}_\e}$ is not guaranteed to be well-defined, i.e. without invoking Hypothesis~\ref{hyp}. \textit{Only} if the additional assumptions at the root of Theorem~\ref{thm:singrwa2} hold, then the two operators coincide and the domain above admits the explicit expression~\eqref{eq:domrwa2}. This may be regarded as an alternative, ``bottom-up'' way to construct the operator $\tilde{H}_{f,\tilde{\omega}_\e}$, as opposed to the ``top-down'' approach followed beforehand.
\end{remark}
\begin{remark}
	In the case $n=1$, i.e. for the restriction $\tilde{H}^{(1)}_{f,\tilde{\omega}_\e}$ of the model to the single-excitation sector $\hfrak^{(1)}$, the results of Prop.~\ref{prop:fiber} can be further improved by noticing (see Remark~\ref{remark:self2}) that $\tilde{\mathcal{S}}^{(0)}_f(z)$ is a bounded operator. Consequently:
	\begin{itemize}
		\item even in the case $s=2$, $\tilde{H}^{(1)}_{f,\tilde{\omega}_\e}$ is self-adjoint for arbitrary values of $\lambda$;
		\item \textit{norm} resolvent convergence holds,
	\end{itemize}
	so that the results of \cite{facchi2021spectral} for the Friedrichs-Lee model with form factor $f\in\hilb_{-2}$ are completely recovered.
\end{remark}

\section{Concluding remarks}
We have shown that, by constructing scales of Fock spaces, it is possible to define creation and annihilation operators for a non-normalizable function $f\in\hilb_{-s}$ for $s\ge1$, i.e. satisfying a weaker growth constraint. This formalism has been used, as a first simple application, to introduce in a natural way a class of GSB models which allow us, for small enough values of the coupling constant $\lambda$, to select non-normalizable form factors $f_1,\dots,f_r\in\hilb_{-1}$, thus extending considerably the class of physical systems that can be rigorously described by GSB models. 

Furthermore, this result has been improved for a particular instance of such models, namely the rotating-wave (RW) spin-boson model, for which a nonperturbative result has been obtained via an explicit evaluation of its domain and its resolvent. Namely, after addressing the rotating-wave spin-boson model with $f\in\hilb_{-1}$, a further extension to form factors $f\in\hilb_{-2}$ has been obtained by means of a delicate technique closely resembling the renormalization procedures of quantum field theory. In all cases, the ``singular'' models introduced in this work have been shown to include the regular ones as a particular case, and can be approximated by them either in the norm or strong resolvent sense, respectively for form factors in $\hilb_{-1}$ or $\hilb_{-2}$.

We will list here some possible developments of the results presented here. First of all, while the findings in Section~\ref{sec:singgsb} for the GSB models are perturbative (and thus valid for small enough values of the coupling constant $\lambda$), the standard GSB models are known to be self-adjoint, as long as the atomic operators are bounded, for arbitrary values of $\lambda$~\cite{arai1997existence}. By assuming $f\in\hilb_{-s}$ for $s\in[0,1]$, a sharper estimate on the maximum value of $\lambda$ for which self-adjointness is ensured, probably dependent on $s$ and converging to $\infty$ when $s\to0$, may be obtained. Correspondingly, an extension to our results to unbounded atomic operators could also be achieved. Likewise, an application of the formalism developed here to the other two paradigmatic models discussed in this paper, cf. Eq.~\eqref{eq:def_sb} and~\eqref{eq:def_sb_pd}, or to the multi-atom generalization of the rotating-wave spin-boson model, see Remark~\ref{remark:natom}, is desirable and will be the object of future research.

Furthermore, while for simplicity we have only dealt with boson field with positive mass cutoff $m>0$, extending our results to the case of a massless boson field $m=0$ should be feasible; physically, this entails to taking into account both infrared and ultraviolet divergences. Finally, the formalism developed here may be also applied to more sophisticated models beyond the GSB structure, for instance involving quadratic terms in $\a{f}$ and $\adag{f}$~\cite{teranishi2015self}.

\section*{Acknowledgments}
We acknowledge fruitful discussions with Paolo Facchi. This work is partially supported by Istituto Nazionale di Fisica Nucleare (INFN) through the project “QUANTUM” and by the Italian National Group of Mathematical Physics (GNFM-INdAM).

\section*{Data availability statement}
Data sharing is not applicable to this article as no new data were created or analyzed in this study.

\AtNextBibliography{\small}
\DeclareFieldFormat{pages}{#1}
\printbibliography

\end{document}